\RequirePackage{amsmath}

\documentclass{llncs}

\usepackage{amsfonts}
\usepackage{todonotes}
\usepackage{paralist}
\usepackage{etoolbox}

\usepackage{thmtools}
\usepackage{thm-restate}
\usepackage{mathtools}

\DeclareMathOperator{\tops}{top}
\DeclareMathOperator{\bots}{bot}
\DeclareMathOperator{\cross}{cr}
\DeclareMathOperator{\pcr}{cr^\star}

\let\doendproof\endproof
\renewcommand{\endproof}{\qed\doendproof}

\renewcommand{\implies}{\Rightarrow}
\renewcommand{\iff}{\Leftrightarrow}

\newcommand{\eps}{\varepsilon}
\newcommand{\comment}[1]{}
\renewcommand{\phi}{\varphi}

\newcommand{\myparagraph}[1]{\smallskip\noindent\textbf{#1.}}
\renewcommand{\l}{\ensuremath{l}}

\title{Level Planarity: Transitivity vs.\ Even Crossings}

\author{Guido Br\"uckner\inst{1}
   \and Ignaz Rutter\inst{2}
   \and Peter Stumpf\inst{2}}

\institute{Karlsruhe Institute of Technology
      \and University of Passau}

\begin{document}

\maketitle

\begin{abstract}
    Recently, Fulek et al.~\cite{fps-htfrp-15,fps-htfrp2-16,fpss-htmdalp-13}
have presented Hanani-Tutte results for (radial) level planarity, i.e., a graph
is (radial) level planar if it admits a (radial) level drawing where any two
(independent) edges cross an even number of times.
    We show that the \textsc{2-Sat} formulation of level planarity testing due
to Randerath et al.~\cite{rsbhkmsc-asfopolg-01} is equivalent to the strong
Hanani-Tutte theorem for level planarity~\cite{fpss-htmdalp-13}. 
    Further, we show that this relationship carries over to radial level planarity, which yields a novel polynomial-time algorithm for testing radial level planarity.
\end{abstract}

\section{Introduction}
\label{sec:intro}

Planarity of graphs is a fundamental concept for graph theory as a whole, and for graph drawing in particular.
Naturally, variants of planarity tailored specifically to directed graphs have been explored.
A planar drawing is \emph{upward planar} if all edges are drawn as monotone curves in the upward direction.
A special case are \emph{level planar drawings} of level graphs, where the input graph~$G=(V,E)$ comes with a level assignment~$\ell \colon V \to \{1, 2, \ldots, k\}$ for some~$k \in \mathbb N$ that satisfies~$\ell(u) < \ell(v)$ for all~$(u, v) \in E$.
One then asks whether there is an upward planar drawing such that each vertex~$v$ is mapped to a point on the horizontal line~$y = \ell(v)$ representing the level of~$v$.
There are also radial variants of these concepts, where edges are drawn as curves that are monotone in the outward direction in the sense that a curve and any circle centered at the origin intersect in at most one point.
Radial level planarity is derived from level planarity by representing levels as concentric circles around the origin.

Despite the similarity, the variants with and without levels differ significantly in their complexity.
Whereas testing upward planarity and radial planarity are
NP-complete~\cite{gt-otccouarpt-02}, level planarity and radial level planarity
can be tested in polynomial time.
In fact, linear-time algorithms are known for both
problems~\cite{bbf-rlptaeilt-05,jl-lpeilt-99}.
However, both algorithms are quite complicated, and subsequent research has led
to slower but simpler algorithms for these problems~\cite{hh-plptalwec-08,rsbhkmsc-asfopolg-01}.  Recently also constrained variants of the level planarity problem have been considered~\cite{br-pclp-17,kr-olp-17}.

One of the simpler algorithms is the one by Randerath et al.~\cite{rsbhkmsc-asfopolg-01}.
It only considers proper level graphs, where each edge connects vertices on adjacent levels.
This is not a restriction, because each level graph can be subdivided to make it proper, potentially at the cost of increasing its size by a factor of~$k$.
It is not hard to see that in this case a drawing is fully specified by the vertex ordering on each level.
To represent this ordering, define a set of variables~$\mathcal V = \{ uw \mid u, w \in V, u \neq w, \ell(u) = \ell(w) \}$.
Randerath et al.\ observe that there is a trivial way of specifying the existence of a level-planar drawing by the following consistency~\eqref{eq:randerath-constraint-consistency}, transitivity~\eqref{eq:randerath-constraint-transitivity} and planarity constraints~\eqref{eq:randerath-constraint-planarity}:
\begin{alignat}{5}
    & \forall uw     &&{}\in \mathcal V                                    &&{}:
\quad uw          &&{}\iff \neg && wu \label{eq:randerath-constraint-consistency}  \\
    & \forall uw, wy &&{}\in \mathcal V                                    &&{}:
\quad uw \land wy &&{}{}\implies      &&{} uy \label{eq:randerath-constraint-transitivity} \\
 & \forall uw, vx &&{}\in \mathcal V \text{ with } (u, v), (w, x) \in E\text{ independent} &&{}: \quad uw
&&{}\iff          &&{} vx \label{eq:randerath-constraint-planarity}
\end{alignat}
The surprising result due to Randerath et al.~\cite{rsbhkmsc-asfopolg-01} is
that the satisfiability of this system of constraints (and thus the existence of
a level planar drawing) is equivalent to the satisfiability of a \emph{reduced
constraint system} obtained by omitting the transitivity constraints~\eqref{eq:randerath-constraint-transitivity}.
That is, transitivity is irrelevant for the satisfiability.
Note that a satisfying assignment of the reduced system is not necessarily
transitive, rather Randerath et al.\ prove that a solution can be made transitive without invalidating the other constraints.
Since the remaining conditions~\ref{eq:randerath-constraint-consistency} and~\ref{eq:randerath-constraint-planarity} can be easily expressed in terms of \textsc{2-Sat}, which can be solved efficiently, this yields a polynomial-time algorithm for level planarity.

A very recent trend in planarity research are Hanani-Tutte style results.
The (strong) Hanani-Tutte theorem~\cite{c-uwukidr-34,t-tatocn-70} states that a graph is planar if and only if it can be drawn so that any two independent edges (i.e., not sharing an endpoint) cross an even number of times.
One may wonder for which other drawing styles such a statement is true.
Pach and T\'oth~\cite{pt-mdopg-04,pt-mdopg-11} showed that the weak Hanani-Tutte theorem (which requires even crossings for all pairs of edges) holds for a special case of level planarity and asked whether the result holds in general.
This was shown in the affirmative by Fulek et al.~\cite{fpss-htmdalp-13}, who also established the strong version for level planarity.
Most recently, both the weak and the strong Hanani-Tutte theorem have been established for radial level planarity~\cite{fps-htfrp-15,fps-htfrp2-16}.

\subsubsection*{Contribution.}

We show that the result of Randerath et al.~\cite{rsbhkmsc-asfopolg-01} from 2001 is
equivalent to the strong Hanani-Tutte theorem for level planarity.

The key difference is that Randerath et
al.\ consider proper level graphs, whereas Fulek et al.~\cite{fpss-htmdalp-13} work
with graphs with only one vertex per
level. For a graph~$G$ we define two graphs $G^\star$, $G^+$ that are equivalent
to $G$ with respect to level planarity. We show how to transform a Hanani-Tutte drawing of a graph~$G^\star$
into a satisfying assignment for the constraint system of $G^+$ and vice versa.  Since this transformation does
not make use of the Hanani-Tutte theorem nor of the result by
Randerath et al., this establishes the equivalence of the two
results.

Moreover, we show that the transformation can be adapted also to the
case of radial level planarity.  This results in a novel
polynomial-time algorithm for testing radial level planarity by
testing satisfiability of a system of constraints that, much like the
work of Randerath et al., is obtained from omitting all transitivity
constraints from a constraint system that trivially models radial
level planarity.  Currently, we deduce the correctness of the new
algorithm from the strong Hanani-Tutte theorem for radial level
planarity~\cite{fps-htfrp2-16}.  However, also this transformation
works both ways, and a new correctness proof of our algorithm in the
style of the work of Randerath et al.~\cite{rsbhkmsc-asfopolg-01} may
pave the way for a simpler proof of the Hanani-Tutte theorem for
radial level planarity.  We leave this as future work.
The proofs of lemmas marked with $(\star)$ can be found in the
appendix.

\section{Preliminaries}
\label{sec:preliminaries}

A \emph{level graph} is a directed graph~$G = (V, E)$ together with a \emph{level assignment}~$\ell: V \to \{1, 2, \ldots, k\}$ for some~$k \in \mathbb N$ that satisfies~$\ell(u) < \ell(v)$ for all~$(u, v) \in E$.
If~$\ell(u) + 1 = \ell(v)~$ for all~$(u, v) \in E$, the level graph~$G$ is \emph{proper}.
Two independent edges~$(u, v), (w, x)$ are \emph{critical} if~$\ell(u) \le \ell(x)$ and~$\ell(v) \ge \ell(w)$.
Note that any pair of independent edges that can cross in a level drawing of $G$ is a pair of critical edges.
Throughout this paper, we consider drawings that may be non-planar, but we assume at all times that no two distinct vertices are drawn at the exact same point, no edge passes through a vertex, and no three (or more) edges cross in a single point.
If any two independent edges cross an even number of times in a drawing~$\Gamma$ of~$G$, it is called a \emph{Hanani-Tutte drawing} of~$G$.

For any~$k$-level graph~$G$ we now define a \emph{star form}~$G^\star$
so that every level of~$G^\star$ consists of exactly one vertex.  The
construction is similar to the one used by Fulek et
al.~\cite{fpss-htmdalp-13}.  Let~$n_i$ denote the number of vertices
on level~$i$ for~$1 \le i \le k$.  Further,
let~$v_1, v_2, \ldots, v_{n_i}$ denote the vertices on level~$i$.
Subdivide every level~$i$ into~$2n_i$
sublevels~$1^i, 2^i, \ldots, (2n_i)^i$.  For~$1 \le j \le n_i$, replace
vertex~$v_j$ by two vertices~$v_j'$, $v_j''$ with~$\ell(v_j') = j^i$
and~$\ell(v_j'') = n_i + j^i$ and connect them by an
edge~$(v_j', v_j'')$, referred to as the \emph{stretch edge}~$e(v_j)$.
Connect all incoming edges of~$v_j$ to~$v_j'$ instead and connect all
outgoing edges of~$v_j$ to~$v_j''$ instead.  Let~$e = (u, v)$ be an
edge of~$G$.  Then let~$e^\star$ denote the edge of~$G^\star$ that
connects the endpoint of~$e(u)$ with the starting point of~$e(v)$.
See Figure~\ref{fig:normalize}.  Define $G^+$ as
the graph obtained by subdividing the edges of~$G^\star$ so that the
graph becomes proper; again, see Figure~\ref{fig:normalize}.
Let~$(u, v), (w, x)$ be critical edges in~$G^\star$.  Define their
\emph{limits} in~$G^+$ as~$(u', v'), (w', x')$ where~$u', v'$ are
endpoints or subdivision vertices of~$(u, v)$,~$w', x'$ are endpoints
or subdivision vertices of~$(w, x)$ and it
is~$\ell(u') = \ell(w') = \max(\ell(u), \ell(w))$
and~$\ell(v') = \ell(x') = \min(\ell(v), \ell(x))$.

\begin{figure}[t]
    \centering
    \includegraphics[width=\linewidth]{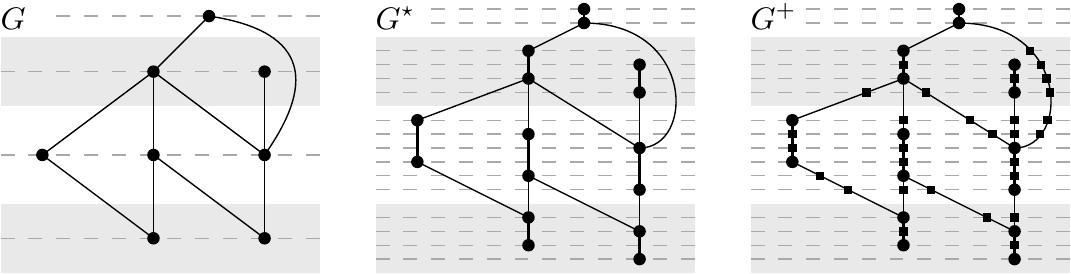}
    \caption{
        A level graph~$G$ (a) modified to a graph~$G^\star$ so as to have only one vertex per level (b) and its proper subdivision~$G^+$ (c).
    }
    \label{fig:normalize}
\end{figure}

\begin{restatable}[$\star$]{lemma}{normalizationlemma}
\label{lem:normalization-lemma}
Let~$G$ be a level graph.  Then

    \smallskip
    \noindent\centering
    $G$ is (radial) level pl.~$\iff G^\star$ is (radial) level pl.~$\iff G^+$ is (radial) level pl.\
\end{restatable}

\section{Level Planarity}
\label{sec:level-planarity}

Recall from the introduction that Randerath et al.\ formulated level
planarity of a proper level graph $G$ as a Boolean satisfiability
problem $\mathcal S'(G)$ on the
variables~$\mathcal V = \{ uw \mid u \neq w, \ell(u) = \ell(w) \}$ and
the clauses given by
Eq.~(\ref{eq:randerath-constraint-consistency})--(\ref{eq:randerath-constraint-planarity}).

It is readily observed that~$G$ is level planar if and only
if~$\mathcal S'(G)$ is satisfiable.  Now let~$\mathcal S(G)$ denote
the \textsc{Sat} instance obtained by removing the transitivity
clauses~(\ref{eq:randerath-constraint-transitivity})
from~$\mathcal S'(G)$.  Note that it
is~$(uw \implies \neg wu) \equiv (\neg uw \lor \neg wu)$
and~$(uw \implies vx) \equiv (\neg uw \lor vx)$, i.e.,~$\mathcal S(G)$
is an instance of 2-\textsc{Sat}, which can be solved efficiently.
The key claim of Randerath et al.\ is that~$\mathcal S'(G)$ is
satisfiable if and only if~$\mathcal S(G)$ is satisfiable, i.e.,
dropping the transitivity clauses does not change the satisfiability
of~$\mathcal S'(G)$.  In this section, we show that $\mathcal S(G)$ is
satisfiable if and only if $G^\star$ has a Hanani-Tutte level drawing
(Theorem~\ref{thm:hanani-tutte-iff-satisfiability}).  Of course, we do
not use the equivalence of both statements to level planarity of $G$.
Instead, we construct a satisfying truth assignment of $\mathcal S(G)$
directly from a given Hanani-Tutte level drawing
(Lemma~\ref{lem:hanani-tutte-implies-satisfiability}), and vice versa
(Lemma~\ref{lem:satisfiability-implies-hanani-tutte}).  This directly
implies the equivalence of the results of Randerath et al.\ and Fulek
et al.~(Theorem~\ref{thm:hanani-tutte-iff-satisfiability}).

The common ground for our constructions is the constraint system
$\mathcal S'(G^+)$, where a Hanani-Tutte drawing implies a variable
assignment that does not necessarily satisfy the planarity
constraints~(\ref{eq:randerath-constraint-planarity}), though in a controlled
way, whereas a satisfying assignment of~$\mathcal S(G)$ induces an
assignment for~$\mathcal S'(G^+)$ that satisfies the planarity
constraints but not the transitivity constraints~(\ref{eq:randerath-constraint-transitivity}).  Thus, in a sense,
our transformation trades planarity for transitivity
and vice versa.

A (not necessarily planar) drawing~$\Gamma$ of~$G$ \emph{induces} a truth assignment~$\varphi$ of~$\mathcal V$ by defining for all~$uw \in \mathcal V$ that~$\varphi(uw)$ is true if and only if~$u$ lies to the left of~$w$ in~$\Gamma$.
Note that this truth assignment must satisfy the consistency clauses, but does not necessarily satisfy the planarity constraints.
The following lemma describes a relationship between certain truth assignments of $\mathcal S(G)$ and crossings in $\Gamma$ that we use to prove Lemmas~\ref{lem:hanani-tutte-implies-satisfiability} and~\ref{lem:satisfiability-implies-hanani-tutte}.

\begin{lemma}
    Let~$(u, v), (w, x)$ be two critical edges of~$G^\star$ and let~$(u', v'), (w', x')$ be their limits in~$G^+$.
    Further, let~$\Gamma^\star$ be a drawing of~$G^\star$, let~$\Gamma^+$ be the drawing of~$G^+$ induced by~$\Gamma^\star$ and let~$\varphi^+$ be the truth assignment of~$\mathcal S(G^+)$ induced by~$\Gamma^+$.
    Then~$(u, v)$ and~$(w, x)$ intersect an even number of times in~$\Gamma^\star$ if and only if~$\varphi^+(u'w') = \varphi^+(v'x')$.
    \label{lem:consistent-limits}
\end{lemma}

\begin{proof}
    We may assume without loss of generality that any two edges cross at most once between consecutive levels by introducing sublevels if necessary.
    Let~$X$ be a crossing between~$(u, v)$ and~$(w, x)$ in~$G^\star$; see Fig.~\ref{fig:hanani-tutte-implies-satisfiability}~(a).
    Further, let~$u_1, w_1$ and~$u_2, w_2$ be the subdivision vertices of~$(u, v)$ and~$(w, x)$ on the levels directly below and above~$X$ in~$G^\star$, respectively.
    It is~$\varphi^+(u_1w_1) = \neg \varphi^+(u_2w_2)$.
    In the reverse direction,~$\varphi^+(u_1w_1) = \neg \varphi^+(u_2w_2)$ implies that~$(u, v)$ and~$(w, x)$ cross between the levels~$\ell(u_1)$ and~$\ell(u_2)$.
    Due to the definition of limits, any crossing between~$(u, v)$ and~$(w, x)$ in~$G^\star$ must occur between the levels~$\ell(u') = \ell(w')$ and~$\ell(v') = \ell(x')$.
    Therefore, it is~$\varphi^+(u'w') = \varphi^+(v'x')$ if and only if~$(u, v)$ and~$(w, x)$ cross an even number of times.
\end{proof}

\begin{figure}[t]
    \centering
    \includegraphics{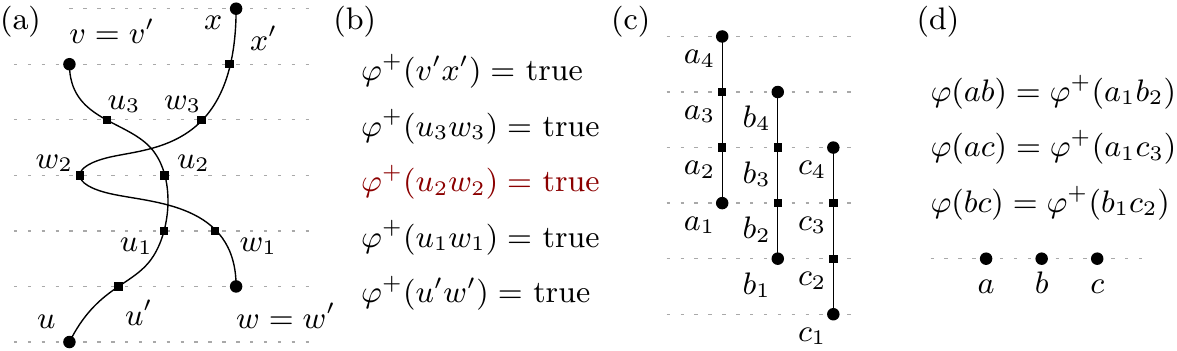}
    \caption{
        A Hanani-Tutte drawing~(a) induces a truth assignment~$\varphi^+$
        that satisfies~$\mathcal S(G^+)$~(b), the value where $\phi^+$ differs
        from $\psi^+$ is highlighted in red.
        Using the subdivided stretch edges of~$G^+$~(c), translate~$\varphi^+$ to a satisfying assignment~$\varphi$ of~$\mathcal S(G)$~(d).
    }
    \label{fig:hanani-tutte-implies-satisfiability}
\end{figure}
\begin{lemma}
    Let~$G$ be a proper level graph and let~$\Gamma^\star$ be a Hanani-Tutte drawing of~$G^\star$.
    Then~$\mathcal S(G)$ is satisfiable.
    \label{lem:hanani-tutte-implies-satisfiability}
\end{lemma}
\begin{proof}
    Let~$\Gamma^+$ be the drawing of~$G^+$ induced by~$\Gamma^\star$ and
    let~$\psi^+$ denote the truth assignment induced by~$\Gamma^+$.
    Note that~$\psi^+$ does not necessarily satisfy the crossing clauses.
    Define~$\varphi^+$ so that it satisfies all clauses of~$\mathcal S(G^+)$ as follows.
  
    Let $u'',w''$ be two vertices of $G^+$ with $\ell(u'')=\ell(w'')$. 
    If one of them is a vertex in $G^\star$, then set
    $\phi^+(u'',w'')=\psi^+(u'',w'')$.
    Otherwise $u''$, $w''$ are subdivision vertices of two edges
$(u,v),(w,x)\in E(G^\star)$. 
    If they are independent, then they are critical. In that case their limits
$(u',v'),(w',x')$ are already assigned consistently by
Lemma~\ref{lem:consistent-limits}.  
    Then set~$\varphi^+(u''w'') = \psi^+(u'w')$.
    If $(u,v)$, $(w,x)$ are adjacent, then we have $u=w$ or $v=x$.  %
    In the first case, we set~$\varphi^+(u''w'') = \psi^+(v'x')$. In the
second case, we set $\varphi^+(u''w'') = \psi^+(u'w')$.

    Thereby, we have for any critical pair of edges $(u'',v''),(w'',x'')\in
E(G^+)$ that $\varphi^+(u''w'')= \varphi^+(v''x'')$ and clearly~$\varphi^+(u''w'')=\neg\varphi^+(w''u'')$. 
    Hence, assignment~$\varphi^+$ satisfies~$\mathcal S(G^+)$.
    See Fig.~\ref{fig:hanani-tutte-implies-satisfiability} for a drawing~$\Gamma^+$ (a) and the satisfying assignment of~$\mathcal S(G^+)$ derived from it (b).

    Proceed to construct a satisfying truth assignment~$\varphi$ of~$\mathcal S(G)$ as follows.
    Let~$u$, $w$ be two vertices of~$G$ with~$\ell(u) = \ell(w)$.
    Then the stretch edges~$e(u), e(w)$ in~$G^\star$ are critical by construction.
    Let~$(u', u''), (w', w'')$ be their limits in~$G^+$.
    Set~$\varphi(uw) = \varphi^+(u'w')$.
    Because~$\varphi^+$ is a satisfying assignment, all crossing clauses of~$\mathcal S(G^+)$ are satisfied, which implies~$\varphi^+(u'w') = \varphi^+(u''w'')$.
    The same is true for all subdivision vertices of~$e(u)$ and~$e(w)$ in~$G^+$.
    Because~$\varphi^+$ also satisfies the consistency clauses of~$\mathcal S(G^+)$, this means that~$\varphi$ satisfies the consistency clauses of~$\mathcal S(G)$.
    See Fig.~\ref{fig:hanani-tutte-implies-satisfiability} for how~$\mathcal S(G^+)$ is translated from~$G^+$ (c) to~$G$ (d).
    Note that the resulting assignment is not necessarily transitive, e.g., it could be~$\varphi(uv) = \varphi(vw) = \neg \varphi(uw)$.

    Consider two edges~$(u, v)$, $(w, x)$ in~$G$ with~$\ell(u) = \ell(w)$.
    Because $G$ is proper, we do not have to consider other pairs of edges.
    Let~$(u', u''), (w', w'')$ be the limits of~$e(u), e(w)$ in~$G^+$.
    Further, let~$(v', v''), (x', x'')$ be the limits of~$e(v), e(x)$ in~$G^+$.
    Because there are disjoint directed paths from~$u'$ and~$w'$ to~$v'$ and~$x'$ and~$\varphi^+$ is a satisfying assignment, it is~$\varphi^+(u'w') = \varphi^+(v'x')$.
    Due to the construction of~$\varphi$ described in the previous paragraph, this means that it is~$\varphi(uw) = \varphi(vx)$.
    Therefore,~$\varphi$ is a satisfying assignment of~$\mathcal S(G)$.
\end{proof}
\begin{lemma}
    Let~$G$ be a proper level graph together with a satisfying truth assignment~$\varphi$ of~$\mathcal S(G)$.
    Then there exists a Hanani-Tutte drawing~$\Gamma^\star$ of~$G^\star$.
    \label{lem:satisfiability-implies-hanani-tutte}
\end{lemma}
\begin{proof}
    We construct a satisfying truth assignment~$\varphi^+$ of~$\mathcal S(G^+)$ from~$\varphi$ by essentially reversing the process described in the proof of Lemma~\ref{lem:hanani-tutte-implies-satisfiability}.
    Proceed to construct a drawing~$\Gamma^+$ of~$G^+$ from~$\varphi^+$ as follows.
    Recall that by construction, every level of~$G^+$ consists of exactly one non-subdivision vertex.
    Let~$u$ denote the non-subdivision vertex of level~$i$.
    Draw a subdivision vertex~$w$ on level~$i$ to the right of~$u$ if~$\varphi^+(uw)$ is true and to the left of~$u$ otherwise.
    The relative order of subdivision vertices on either side of~$u$ can be chosen arbitrarily.
    Let~$\Gamma^\star$ be the drawing of~$G^\star$ induced by~$\Gamma^+$.
    To see that~$\Gamma^\star$ is a Hanani-Tutte drawing, consider two critical edges~$(u, v), (w, x)$ of~$G^\star$.
    Let~$(u', v'), (w', x')$ denote their limits in~$G^+$.
    One vertex of~$u'$ and~$v'$ ($w'$ and~$x'$) is a subdivision vertex and the other one is not.
    Lemma~\ref{lem:consistent-limits} gives~$\varphi^+(u'w') = \varphi^+(v'x')$ and then by construction~$u', w'$ and~$v', x'$ are placed consistently on their respective levels.
    Moreover, Lemma~\ref{lem:consistent-limits} yields that~$(u, v)$ and~$(w, x)$ cross an even number of times in~$G^\star$.
\end{proof}

\begin{figure}[t]
    \centering
    \includegraphics{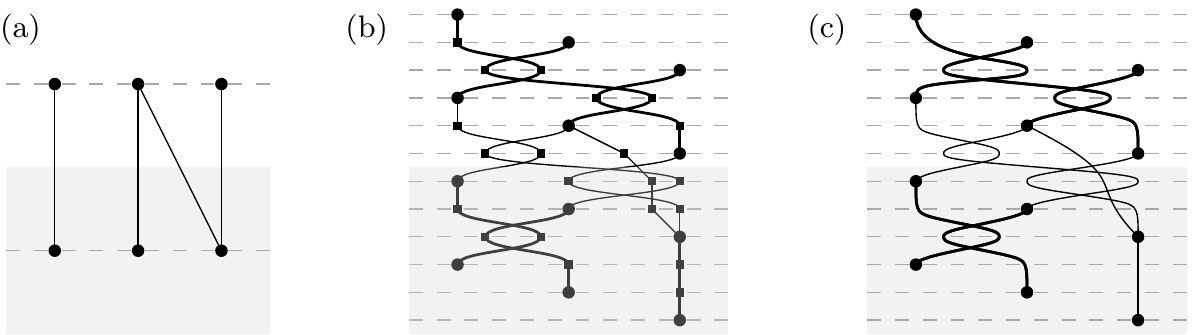}
    \caption{
        A proper level graph~$G$ together with a satisfying variable assignment~$\varphi$ (a) induces a drawing of~$G^+$ (b), which induces a Hanani-Tutte drawing of~$G^\star$ (c).
    }
    \label{fig:satisfiability-implies-hanani-tutte}
\end{figure}

\begin{theorem}
    Let $G$ be a proper level graph.
    Then

    \smallskip
    \noindent\centering
    $\mathcal S(G)$ is satisfiable $\iff G^\star$ has a Hanani-Tutte level drawing $\iff G$ is level planar.
    \label{thm:hanani-tutte-iff-satisfiability}
\end{theorem}
\section{Radial Level Planarity}
\label{sec:radial}

In this section we present an analogous construction for radial level
planarity.  In contrast to level planarity, we now have to consider
cyclic orders on the levels, and even those may still leave some
freedom for drawing the edges between adjacent levels.  In the
following we first construct a constraint system of radial level
planarity for a proper level graph~$G$, which is inspired by the one
of Randerath et al.  Afterwards, we slightly modify the construction
of the modified graph~$G^\star$.  Finally, in analogy to the level planar
case, we show that a satisfying assignment of our constraint
system defines a satisfying assignment of the constraint system
of~$G^+$, and that this in turn corresponds to a Hanani-Tutte radial level drawing
of~$G^\star$.

\myparagraph{A Constraint System for Radial Level Planarity}
We start with a special case that bears a strong similarity with the
level-planar case.  Namely, assume that~$G$ is a proper level graph that contains a directed path
$P = \alpha_1,\dots,\alpha_k$ that has exactly one
vertex~$\alpha_i$ on each level~$i$.  We now express the cyclic
ordering on each level as linear orders whose first vertex is
$\alpha_i$.  To this end, we introduce for each level the
variables~$\mathcal V_i = \{\alpha_i uv \mid u,v \in V_i \setminus
\{\alpha_i\} \}$, where $\alpha_iuv\equiv \text{true}$ means $\alpha_i$, $u$,
$v$ are arranged clockwise on the circle representing level $i$.
We further impose the following necessary and sufficient \emph{linear
  ordering constraints}~$\mathcal L_G(\alpha_i)$.
\begin{alignat}{6}
    & \forall \text{          distinct } && u, v    &&\in V \setminus \{\alpha_i\}: \quad \alpha_iuv                  &&\iff \neg && \alpha_ivu     \label{eq:radial-constraint-consistency}  \\
    & \forall \text{ pairwise distinct } && u, v, w &&\in V \setminus \{\alpha_i\}: \quad \alpha_iuv \land \alpha_ivw &&\implies \neg && \alpha_iuw \label{eq:radial-constraint-transitivity}
\end{alignat}
It remains to constrain the cyclic orderings of vertices on adjacent
levels so that the edges between them can be drawn without crossings.
For two adjacent levels~$i$ and~$i+1$,
let~$\eps_i = (\alpha_i,\alpha_{i+1})$ be the \emph{reference edge}.
Let~$E_i$ be the set of edges~$(u,v)$ of~$G$ with~$\ell(u) = i$ that
are not adjacent to an endpoint of~$\eps_i$.  Further~$E_i^+ = \{(\alpha_i,v) \in E \setminus \{\eps_i\} \}$
and~$E_{i}^- = \{(u,\alpha_{i+1}) \in E \setminus \{\eps_i\}\}$
denote the edges between levels~$i$ and~$i+1$ adjacent to the
reference edge~$\eps_i$.

In the context of the constraint formulation, we only consider drawings of the edges between levels~$i$ and~$i+1$ where any pair of edges crosses at most once and, moreover,~$\eps_i$ is not crossed.
Note that this can always be achieved, independently of the orderings chosen for levels~$i$ and~$i+1$.
Then, the cyclic orderings of the vertices on the levels~$i$ and~$i+1$ determine the drawings of all edges in~$E_i$.
In particular, two edges~$(u,v)$, $(u',v') \in E_i$ do not intersect if and only if~$\alpha_iuu' \Leftrightarrow \alpha_{i+1}vv'$; see Fig.~\ref{fig:radial-sat}~(a).
Therefore, we introduce constraint~\eqref{eq:radial-constraint-planarity-one}.
For each edge~$e$ in~$E_i^+ \cup E_i^-$ it remains to decide whether it is embedded locally to the left or to the right of~$\eps_i$.
We write~$\l(e)$ in the former case.
For any two edges~$e \in E_i^-$ and~$f \in E_i^+$ we have that~$e$ and~$f$ do not cross if and only if~$\l(e) \Leftrightarrow \neg \l(f)$; see Fig.~\ref{fig:radial-sat}~(b).
This gives us constraint~\eqref{eq:radial-constraint-planarity-two}.
It remains to forbid crossings between edges in~$E_i$ and edges in~$E_i^+ \cup E_i^-$.
An edge~$e = (\alpha_i,v'') \in E_i^+$ and an edge~$(u',v') \in E_i$ do not cross if and only if~$\l(e) \Leftrightarrow \alpha_{i+1}v'v''$; see Fig.~\ref{fig:radial-sat}~(c).
Crossings with edges~$(v,\alpha_{i+1}) \in E_i^-$ can be treated analogously.
This yields constraints~\eqref{eq:radial-constraint-planarity-three} and~\ref{eq:radial-constraint-planarity-four}.
We denote the planarity
constraints~\eqref{eq:radial-constraint-planarity-one}--\eqref{eq:radial-constraint-planarity-four} by~$\mathcal P_G(\eps_i)$, where~$\eps_i = (\alpha_i,\alpha_{i+1})$.
\begin{alignat}{5}%
    & \forall \text{ independent } (u, v), (u', v') \in E_i                       &&:\quad \alpha_iuu'            &&\iff      \alpha_{i+1}vv'  \label{eq:radial-constraint-planarity-one}   \\
    & \forall                      e \in E_i^+, f \in E_i^-                       &&:\quad l(e)                     &&\iff \neg \l(f)              \label{eq:radial-constraint-planarity-two}   \\
    & \forall \text{ independent } (\alpha_i,v'') \in E_i^+, (u,v) \in E_i      &&:\quad \l(\alpha_i,v'')     &&\iff      \alpha_{i+1}vv'' \label{eq:radial-constraint-planarity-three} \\
    & \forall \text{ independent } (u'', \alpha_{i+1}) \in E_i^-, (u,v) \in E_i &&:\quad \l(u'',\alpha_{i+1}) &&\iff      \alpha_{i}uu''   \label{eq:radial-constraint-planarity-four}
\end{alignat}%
\begin{figure}[t]%
    \centering
    \includegraphics[page=2]{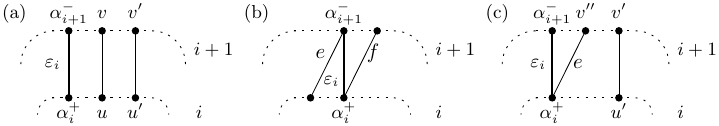}
    \caption{
        Illustration of the planarity constraints for radial planarity for the
case of two edges in~$E_i$ (a),
constraint~\eqref{eq:radial-constraint-planarity-one}; the case of an edge in~$e
\in E_i^-$ and an edge~$f \in E_i^+$ (b),
constraint~\eqref{eq:radial-constraint-planarity-two}; and the case of an edge
in~$E_i$ and an edge~$e \in E_i^+$ (c), constraint~\eqref{eq:radial-constraint-planarity-three}.
     }
     \label{fig:radial-sat}
\end{figure}%
It is not difficult to see that the transformation between
Hanani-Tutte drawings and solutions of the constraint system without
the transitivity constraints~\eqref{eq:radial-constraint-transitivity} can be performed as in
the previous section.  The only difference is that one has to deal
with edges that share an endpoint with a reference~$\eps_i$.

In general, however, such a path~$P$ from level~$1$ to level~$k$ does
not necessarily exist.  Instead, we use an arbitrary reference edge between any
two consecutive levels. More formally, we call a pair of sets $A^+=\{\alpha_1^+,\dots, \alpha_k^+\}$,
$A^-=\{\alpha_1^-,\dots, \alpha_k^-\}$ \emph{reference sets} for $G$ if we have
$\alpha_1^-=\alpha_1^+$ and $\alpha_{k}^+=\alpha_k^-$ and for $1\le i \le k$ the \emph{reference vertices} $\alpha_i^+$,
$\alpha_i^-$ lie on level $i$ and for $1\le i < k$ graph $G$ contains the
\emph{reference edge} $\eps_i=(\alpha_i^+,\alpha_{i+1}^-)$ unless there is no
edge between level $i$ and level $i+1$ at all. 
In that case, we can extend
every radial drawing of $G$ by the edge $(\alpha_i^+,\alpha_{i+1}^-)$ without
creating
new crossings. We may therefore assume that this case does not occur and we do
so from now on.

To express radial level planarity, we express the cyclic orderings on
each level twice, once with respect to the reference
vertex~$\alpha_i^+$ and once with respect to the reference
vertex~$\alpha_i^-$.  To express planarity between adjacent levels,
we use the planarity constraints with respect to the reference
edge~$\eps_i$.  It only remains to specify that,
if~$\alpha_i^+ \ne \alpha_i^-$, the linear ordering with respect to
these reference vertices must be linearizations of the same cyclic
ordering.  This is expressed by the following \emph{cyclic ordering
  constraints}~$\mathcal C_G(\alpha_i^+, \alpha_i^-)$.
\begin{alignat}{4}
    & \forall \text{ distinct } u, && \, v \in V_i \setminus \{\alpha_i^-, \alpha_i^+\}: \enspace & (\alpha_i^-uv \iff \alpha_i^+uv) &\iff (\alpha_i^-u\alpha_i^+ \iff \alpha_i^-v\alpha_i^+) \label{eq:radial-constraint-merge-one} \\
    & \forall                      && \, v \in V_i \setminus \{\alpha_i^-, \alpha_i^+\}: \enspace &            \alpha_i^-v\alpha_i^+ &\iff \alpha_i^+\alpha_i^-v                              \label{eq:radial-constraint-merge-two}
\end{alignat}
 The constraint
set~$\mathcal S'(G,A^+,A^-)$ consists of the linearization
constraints~$\mathcal L_G(\alpha_i^+)$ and~$\mathcal L_G(\alpha_i^-)$
and the cyclic ordering
constraints~$\mathcal C_G(\alpha_i^+, \alpha_i^-)$
for~$i = 1, 2, \ldots, k$ if~$\alpha_i^+ \ne \alpha_i^-$, plus the
planarity constraints~$\mathcal P_G(\eps_i)$
for~$i = 1, 2, \ldots, k - 1$.  This completes the definition of our
constraint system, and we proceed to show its correctness.

\begin{restatable}[$\star$]{theorem}{radialconstraints}
	\label{the:radialconstraints}
  Let $G$ be a proper level graph with reference sets $A^+$, $A^-$. 
  Then the constraint system $\mathcal S'(G,A^+,A^-)$ is satisfiable if and
  only if~$G$ is radial level planar.  Moreover, the radial level
  planar drawings of~$G$ correspond bijectively to the satisfying
  assignments of~$\mathcal S'(G,A^+,A^-)$.
\end{restatable}

\noindent
Similar to Section~\ref{sec:level-planarity}, we now define a
reduced constraints system~$\mathcal S(G,A^+,A^-)$ obtained
from~$\mathcal S'(G, A^+, A^-)$ by dropping
constraint~\eqref{eq:radial-constraint-transitivity}.  Observe that this
reduced system can be represented as a system of linear equations
over~$\mathbb F_2$, which can be solved efficiently.  Our main result
is that~$S(G,A^+,A^-)$ is satisfiable if and only if~$G$ is radial
level planar.

\myparagraph{Modified Star Form}
We also slightly modify the splitting and perturbation operation in the construction of the star form $G^\star$ of $G$ for each level $i$.
This is necessary since we need a special treatment of the reference vertices~$\alpha_i^+$ and~$\alpha_i^-$ on each level $i$.
Consider the level $i$ containing the $n_i$ vertices $v_1,\dots,v_{n_i}$.
\begin{figure}[tb]
  \centering
  \includegraphics{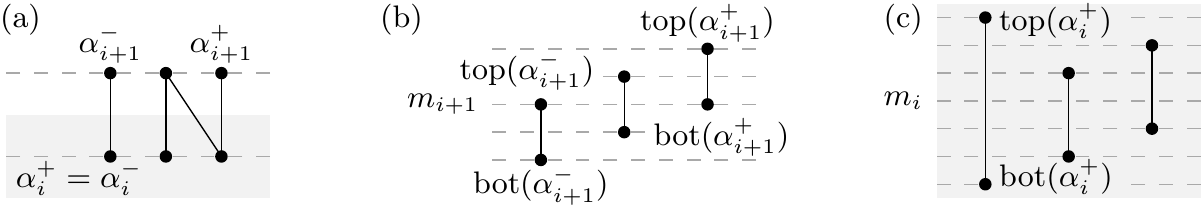}
  \caption{Illustration of the modified construction of the stretch
    edges for $G^\star$ for the graph $G$ in (a). The stretch edges
    for level $i+1$ where~$\alpha_{i+1}^+ \ne \alpha_{i+1}^-$ (b) and
    for level $i$ where~$\alpha_i^+ = \alpha_i^-$ (c).}
  \label{fig:modified-g-star}
\end{figure}
If~$\alpha_i^+ \ne \alpha_i^-$, then we choose the numbering of the
vertices such that~$v_1 = \alpha_i^-$ and~$v_{n_i} = \alpha_i^+$.  We
replace~$i$ by~$2n_{i}-1$ levels~$1^i,2^i,\dots,(2n_i-1)^i$, which is
one level less than previously.  Similar to before, we replace each
vertex~$v_j$ by two vertices~$\bots(v_j)$ and~$\tops(v_j)$
with~$\ell(\bots(v_j)) = j^i$ and~$\ell(\tops(v_j)) = (n_i - 1 + j)^i$
and the corresponding stretch edge~$(\bots(v_j),\tops(v_j))$; see
Fig.~\ref{fig:modified-g-star}~(b).  This ensures that the construction
works as before, except that the middle level~$m_i = j^{n_i}$ contains
two vertices, namely~$\alpha_i^{+}{}''$ and~$\alpha_i^{-}{}'$.

If, on the other hand,~$\alpha_i^+ = \alpha_i^-$, then we choose~$v_1 = \alpha_i^+$.  But now we replace level~$i$ by~$2n_i+1$
levels~$1^i,\dots,(2n_i+1)^i$.  Replace~$v_1$ by vertices~$\bots(v_1),\tops(v_1)$ with~$\ell(\bots(v_1)) = 1^i$
and~$\ell(\tops(v_1)) = (2n_i+1)^j$.
Replace all other~$v_j$ with vertices~$\bots(v_j),\tops(v_j)$ with~$\ell(\bots(v_j)) = j^i$ and~$\ell(\tops(v_j)) = (n_i+1+j)^i$.
For all~$j$, we add the stretch
edge~$(\bots(v_j),\tops(v_j))$ as before; see
Fig.~\ref{fig:modified-g-star}~(c).  This construction ensures that
the stretch edge of~$\alpha_i^+=\alpha_i^-$ starts in the
first new level~$1^i$ and ends in the last new level~$(2n_i+1)^i$, and
the middle level~$m_i = {n_i+1}^j$ contains no vertex.

As before, we replace each original edge~$(u,v)$ of the input graph~$G$ by the edge~$(\tops(u),\bots(v))$ connecting the upper endpoint of the stretch edge of~$u$ to the lower endpoint of the stretch edge of~$v$.  Observe that the construction preserves the properties that for
each level~$i$ the middle level~$m_i$ of the levels that replace~$i$
intersects all stretch edges of vertices on level~$i$.  Therefore,
Lemma~\ref{lem:normalization-lemma} also holds for this modified
version of~$G^\star$ and its proper subdivision~$G^+$.  For each
vertex~$v$ of~$G$ we use~$e(v) = (\bots(v),\tops(v))$ to denote its
stretch edge.

We define the function~$L$ that maps each level $j$
of~$G^\star$ or~$G^+$ to the level $i$ of~$G$ it replaces.  For
an edge $e$ of~$G^\star$ and a level $i$ that intersects~$e$, we
denote by~$e_i$ the subdivision vertex of~$e$ at level $i$ in~$G^+$.
For two levels $i$ and~$j$ that both intersect an edge $e$
of~$G^\star$, we denote by~$e_i^j$ the path from~$e_i$ to~$e_j$
in~$G^+$.

\myparagraph{Constraint System and Assignment for $\boldsymbol{G^+}$}
We now choose reference sets $B^+$, $B^-$ for~$G^+$ that are based on the
reference sets $A^+$, $A^-$ for~$G$.
Consider a level~$j$ of~$G^\star$ and let~$i=L(j)$ be the corresponding level of~$G$.
For each level~$j$, define two vertices~$\beta_j^+$, $\beta_j^-$. If~$\alpha_i^-=\alpha_i^+$, set~$\beta_j^-=\beta_j^+=e(\alpha_i^-)_j$; see Fig.~\ref{fig:assignment-gplus}~(b).
Otherwise, the choice is based on whether~$j$ is the middle level~$m=m_i$ of the levels~$L^{-1}(i)$ that replace level~$i$ of~$G$, or whether~$j$ lies above or below~$m$.
Choose~$\beta_m^- = \tops(\alpha_i^+)$ and~$\beta_m^+ = \bots(\alpha_{i+1}^-)$.
For~$j<m$,  choose~$\beta_j^-=\beta_j^+=e(\alpha_i^-)_j$ and for~$j>m$, choose~$\beta_j^-=\beta_j^+ = e(\alpha_i^+)_j$; see Fig.~\ref{fig:assignment-gplus}~(c).

\begin{figure}[t]
    \centering
    \includegraphics{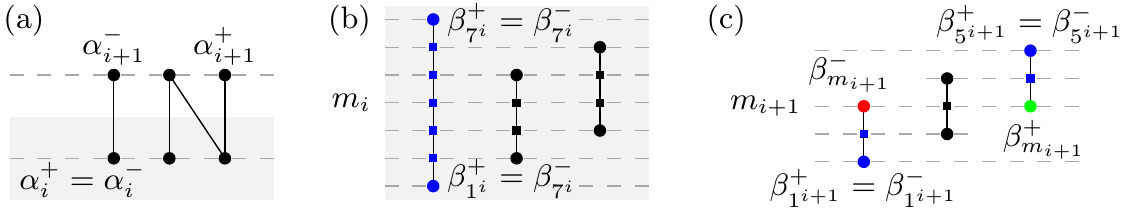}
    \caption{
        Definition of $\beta^+, \beta^-$ in the assignment for $G^+$ for the same graph as in Fig.~\ref{fig:modified-g-star}~(a).
        Vertices $\beta^+$ ($\beta^-$) are drawn in green (red), or in blue if they coincide.
    }
    \label{fig:assignment-gplus}
\end{figure}

We set $B^+$ to be the set containing all $\beta_j^+$ and likewise for~$B^-$.  Our next step is to construct from a satisfying
assignment~$\varphi$ of~$\mathcal S(G,A^+,A^-)$ a corresponding
satisfying assignment~$\varphi^+$ of~$\mathcal S(G^+,B^+,B^-)$.
The
construction follows the approach from Lemma~\ref{lem:satisfiability-implies-hanani-tutte} and makes use of
the fact that $G^+$ is essentially a stretched and perturbed
version of $G$.  Since the construction is straightforward but
somewhat technical, we defer it to Appendix~\ref{sec:phiplus}.

\begin{restatable}[$\star$]{lemma}{existencephiplus}
  If $\mathcal S(G,A^+,A^-)$ is satisfiable, then~$\mathcal S(G^+,B^+,B^-)$ is
satisfiable.
\end{restatable}

\myparagraph{Constructing a Hanani-Tutte Drawing}
We construct a radial drawing~$\Gamma^+$ of~$G^+$, from which we
obtain the drawing~$\Gamma^\star$ of~$G^\star$ by smoothing the subdivision
vertices.  Afterwards we show that~$\Gamma^\star$ is a Hanani-Tutte drawing.

We construct~$\Gamma^+$ as follows.  Consider a level~$j$ of~$G^+$ and
let~$i=L(j)$ be the original level of~$G$.  First assume~$j=m_i$ is
the middle levels of the levels replacing level~$i$ of~$G$.  If~$\beta_j^-=\beta_j^+$, then we place all vertices of~$V_j(G^+)$ in
arbitrary order.  Otherwise, we place $\beta_j^-$ and $\beta_j^+$ arbitrarily on the
circle representing the level~$m_i$.  We then place each vertex~$v\in V_j(G^+)\setminus\{\beta_j^-,\beta_j^+\}$ such that~$\beta_j^-,v,\beta_j^+$ are ordered clockwise if and only if~$\phi(\beta_j^-v\beta_j^+)$ is true (i.e., we place~$v$ on the correct
side of~$\beta_j^-$ and~$\beta_j^+$ and arrange the vertices on both
sides of~$\beta_j^-$ and~$\beta_j^+$ arbitrarily).

Next assume~$j\ne m_i$.  Then there is exactly one vertex~$\xi\in V_j(G^+)\cap V(G^\star)$.  If~$\xi\in B^-$, then we place all
vertices of~$V_j(G^+)$ in arbitrary order on the circle representing
the level~$j$.  Otherwise, we place~$\beta_j^-$ and~$\xi$ arbitrarily.
We then place any vertex~$v\in V_j(G^+)\setminus\{\beta_j^-,\xi\}$
such that~$\beta_j^-,\xi,v$ are ordered clockwise if and only if~$\phi^+(\beta_j^-\xi v)$ is true.  Again, we arrange the vertices on
either side of~$\beta_j^-$ and~$\xi$ arbitrarily.  We have now fixed
the positions of all vertices and it remains to draw the edges.

Consider two consecutive levels~$j$ and~$j+1$ of~$G^+$.  We draw the
edges in~$E_j(G^+)$ such that they do not cross the reference edges in~$E(G^+)\cap(B^+\times B^-)$.  We draw an edge~$e=(\beta_j^+,x')\in E_j^+(G^+)$ such that it is locally left of~$(\beta_j^+,\beta_j^-)$ if and only if~$\phi^+(\l(e))=\text{true}$.
By reversing the subdivisions of the edges in~$G^+$ we obtain~$G^\star$ and
along with that we obtain a drawing~$\Gamma^\star$ of~$G^\star$ from~$\Gamma^+$.

Let $a,b,c$ be curves or corresponding edges. Then we write $\cross(a,b)$ for the
number of crossings between $a,b$ and set $\cross(a,b,c)=\cross(a,b)+\cross(a,c)+\cross(b,c)$.
The following lemma is the radial equivalent to
Lemma~\ref{lem:consistent-limits} and constitutes our main tool for
showing that edges in our drawing cross evenly.

\begin{restatable}[$\star$]{lemma}{threecurves}
  \label{lem:equalOrderImpliesEvenCrossNumberOf3Edges}%
  Let $C_1$ and $C_2$ be distinct concenctric circles and let $a,b,c$
  be radially monotone curves from $C_1$ to $C_2$ with pairwise
  distinct start- and endpoints that only intersect at discrete
  points.  Then the start- and endpoints of $a,b,c$ have the same
  order on $C_1$ and $C_2$ if and only if
  $\cross(a,b,c)\equiv 0 \mod 2$.
\end{restatable}

\begin{lemma}
  The drawing~$\Gamma^\star$ is a Hanani-Tutte drawing of $G^\star$.
\end{lemma}

\begin{proof}
  We show that each pair of {}independent{} edges of~$G^\star$ crosses evenly 
  in~$\Gamma^\star$.  Of course it suffices to consider critical pairs of
  edges, since our drawing is radial by construction, and therefore
  non-critical independent edge pairs cannot cross. 

  Every edge $(\alpha_i^+,\alpha_{i+1}^-)$ is subdivided into edges of
  the form $(\beta_j^+,\beta_{j+1}^-)$ and therefore it is not
  crossed.

  Let $e,f$ be two independent edges in
  $E(G^\star)\setminus (A^+\times A^-)$ that are critical. Let $a$ and
  $b$ be the innermost and outermost level shared by $e$ and~$f$.
	
  We seek to use
  Lemma~\ref{lem:equalOrderImpliesEvenCrossNumberOf3Edges} to analyze
  the parity of the crossings between $e$ and~$f$.  To this end, we
  construct a curve $\gamma$ along the edges of the form
  $(\beta_j^+,\beta_{j+1}^-)$ as follows. For every level $j$ we add a
  curve $c_j$ between $\beta_j^-$ and $\beta_j^+$ on the circle
  representing the level $j$ (a point for $\beta_j^-=\beta_j^+$; chosen
arbitrarily otherwise).  The
  curve $\gamma$ is the union of these curves $c_j$ and the curves for
  the edges of the form $(\beta_j^+,\beta_{j+1}^-)$.  Note that $\gamma$
  spans from the innermost level $1$ to the outermost level $(2n_k+1)^k$ with endpoints
$\bots(\alpha_1^+)$ and $\tops(\alpha_k^-)$.

  For any edge $g\in G^\star$, we denote its curve in $\Gamma^\star$ by
  $c(g)$.  For any radial monotone curve $c$ we denote its subcurve
  between level $i$ and level $j$ by $c_i^j$ (using only one point on
  circle $i$ and circle $j$ each).  We consider the three curves
  $g'=\gamma_a^b,e'=c(e)_a^b,f'=c(f)_a^b$.  We now distinguish cases
  based on whether one of the edges $e,f$ starts at the bottom end or
  ends at the top end of the reference edges on level $a$ or $b$.

  \noindent\textbf{Case 1:} We have $e_a,f_a\ne \beta_a^+$ and $e_b,f_b\ne \beta_b^-$.
  Note that
  $\cross(e,f,\gamma)=\cross(e,f)+\cross(e,\gamma)+\cross(f,\gamma)$,
  and
  therefore~$\cross(e,f) \equiv \cross(e,f,\gamma)+\cross(e,\gamma) +
  \cross(f,\gamma) \mod 2$.
  
  By Lemma~\ref{lem:equalOrderImpliesEvenCrossNumberOf3Edges} we have
  that the orders of $e_a,f_a,\beta_a^+$ and $e_b,f_b,\beta_b^-$
  differ if and only if $\cross(e,f,\gamma) \equiv 1 \mod 2$.  That
  is~$\cross(e,f,\gamma) \equiv 0 \mod 2$ if and only
  if~$\phi^+(\beta_a^+,e_a,f_a) = \phi^+(\beta_b^-,e_b,f_b)$.  We show
  that $\phi^+(\beta_a^+,e_a,f_a) = \phi^+(\beta_b^-,e_b,f_b)$ if and
  only if $\cross(e,\gamma)+\cross(f,\gamma)\equiv 0\mod 2$.  In
  either case,~$\cross(e,f)$ is even.

  Let $a\le j\le b-1$. By construction we have for
  $\beta_j^-\ne\beta_j^+$ and any other vertex $v$ on level $j$, that
  $\beta_j^-,v,\beta_j^+$ are placed clockwise if and only if
  $\phi^+(\beta_j^-,v,\beta_j^+)$ is true.  Further, since $\phi^+$
  satisfies $\mathcal C(\beta_j^+,\beta_j^-)$, we have for any other
  vertex $u$ on level $j$ that $\beta_j^-,u,v$ and $\beta_j^+,u,v$
  have the same order if and only if $\beta_j^-,v,\beta_j^+$ and
  $\beta_j^-,u,\beta_j^+$ have the same order, i.e., if and only if
  $u$ and $v$ lie on the same side of $\beta_j^-$ and
  $\beta_j^+$. This however, is equivalent to
  $\cross(e,c_j)+\cross(f,c_j)\equiv 0\mod 2$.
		
  Since $\phi^+$ satisfies $\mathcal P(\delta_j)$ where $\delta_j = (\beta_j^+,\beta_{j+1}^-)$, we
  have that
  $\phi^+(\beta_j^+,e_j,f_j)=\phi^+(\beta_{j+1}^-,e_{j+1},f_{j+1})$.
  We obtain, that
  $\phi^+(\beta_j^+,e_j,f_j)=\phi^+(\beta_{j+1}^-,e_{j+1},f_{j+1})$
  unless
  $\phi^+(\beta_{j+1}^-,e_{j+1}f_{j+1})\ne
  \phi^+(\beta_{j+1}^+e_{j+1}f_{j+1})$ (which requires
  $\beta_{j+1}^-\ne \beta_{j+1}^+$). This is equivalent to
  $\cross(e,c_{j+1})+\cross(f,c_{j+1})\equiv 1\mod 2$.  Hence, we have
  $\phi^+(\beta_a^+e_af_a)=\phi^+(\beta_b^+e_bf_b)$ if and only if 
  $\sum_{j=a}^{b-1} \cross(c_j,e)+\cross(c_j,f)\equiv 0\mod 2$ (Note
  that $\beta_b^-=\beta_b^+$. ).  Since edges of the form
  $(\beta_j^+,\beta_{j+1}^-)$ are not crossed, this is equivalent to
  $\cross(\gamma,e)+\cross(\gamma,f) \equiv 0 \mod 2$.  Which we aimed
  to show.  By the above argument we therefore find that~$\cross(e,f)$
  is even.

  \noindent\textbf{Case 2.}  We do not have $e_a,f_a\ne \beta_a^+$ and
  $e_b,f_b\ne \beta_b^-$.  For example, assume $e_a=\beta_a^+$; the
  other cases work analogously.  We then have
  $\beta_a^+=\tops(\alpha_i^+)$.  This means $e$ originates from an
  edge in $G$.  Since such edges do not cross middle levels, $g'$ is a
  subcurve of an original edge $\eps_i$.  Especially, we have only
  three vertices per level between $a$ and $b$ that correspond to
  $\gamma,e,f$.

  Let $H\subseteq G^+$ be the subgraph induced by the vertices of
  $(\eps_i)_a^b$, $e_a^b$, $f_a^b$.  Then~$\varphi^+$ satisfies all the
  constraints of $\mathcal S(H,V((\eps_i)_a^b),V((\eps_i)_a^b))$.  However,
  each level of $H$ contains only three vertices, and therefore the
  transitivity constraints are trivially satisfied, i.e.,~$\varphi^+$
  satisfies all the constraints of
  $\mathcal S'(H,V((\eps_i)_a^b),V((\eps_i)_a^b))$.  Thus, by
  Theorem~\ref{the:radialconstraints}, a drawing $\Gamma_H$ of $H$
  according to $\phi^+$ is planar.  I.e., we have
  $\cross_{\Gamma_H}((\eps_i)_a^b,e_a^b,f_a^b) = 0$.  Let $C_a,C_b$ be
  $\eps$-close circles to levels $a$ and $b$, respectively, that lie
  between levels $a$ and $b$.  With
  Lemma~\ref{lem:equalOrderImpliesEvenCrossNumberOf3Edges} we obtain
  that $\eps_i,e,f$ intersect $C_a$ and $C_b$ in the same order.

  Note that $\Gamma^+$ is drawn according to $\phi^+$ in level $a$ and
  in level $b$.  We obtain that the curves for $\eps_i,e,f$ intersect
  $C_a$ in the same order in $\Gamma^+$ and in $\Gamma_H$.  The same
  holds for $C_b$. Hence, the curves intersect $C_a$ and $C_b$ in the
  same order in $\Gamma^+$.  With
  Lemma~\ref{lem:equalOrderImpliesEvenCrossNumberOf3Edges} we have
  $\cross_{\Gamma^+}((\eps_i)_a^b,e_a^b,f_a^b) \equiv 0\mod 2$.  Since
  $\gamma$ is a subcurve of $\eps_i$ and thus not crossed in $\Gamma^+$, this yields
  $\cross_{\Gamma^+}(e_a^b,f_a^b) \equiv 0\mod 2$.
  We thus have shown that any two independent edges have an even
  number of crossings.
\end{proof}

As in the level planar case the converse also holds. 

\begin{restatable}[$\star$]{lemma}{phiFromDrawing}
	\label{lem:HT2TC}
          Let $G^\star$ be a level graph with reference sets $A^+$, $A^-$ for
$G^+$.  If
$G^\star$ admits a Hanani-Tutte drawing, then there exists a satisfying assignment $\varphi$ of $\mathcal S(G^+,A^+,A^-)$.
\end{restatable}

\begin{theorem}
  Let $G$ be a proper level graph with reference sets $A^+$, $A^-$.  Then\\
        \indent $\mathcal S(G, A^+, A^-) \text{ is satisfiable} \iff G^\star \text{ has a Hanani-Tutte radial level drawing}$ \\
        \indent $\phantom{\mathcal S(G, A^+, A^-) \text{ is satisfiable}} \iff G       \text{ is radial level planar.}$
\end{theorem}

\section{Conclusion}

We have established an equivalence of two results on level planarity
that have so far been considered as independent.  The novel connection
has further led us to a new testing algorithm for radial level
planarity.  Can similar results be achieved for level planarity on a
rolling cylinder or on a torus~\cite{addfp-blp-16}?

\newpage
\bibliographystyle{splncs04}
\bibliography{references}

\newpage
\appendix

\section{Omitted Proofs from Section~\ref{sec:preliminaries}}

\normalizationlemma*

\begin{proof}
    The first equivalence is due to Fulek et al.~\cite{fpss-htmdalp-13}.
    The forward direction is trivially true.
    For the reverse direction, the key insight is that for every level~$i$ of~$G$ there is a level~$i'$ of~$G^\star$ so that for each vertex~$v$ with~$\ell(v) = i$ its stretch edge~$e(v)$ crosses level~$i'$ in~$G^\star$.
    The second equivalence is obvious, because~$G^+$ is simply the proper subdivision of~$G^\star$.
\end{proof}

\section{Correctness of the Radial Constraint System}

\radialconstraints*

\begin{proof}
    Clearly, the radial level planar drawings of~$G$ correspond injectively to the satisfying assignments of~$\mathcal S'(G, A^+, A^-)$.
    For the reverse direction, consider a satisfying assignment~$\varphi$ of~$\mathcal S'(G, A^+, A^-)$.
    Start by observing that the
constraints~\eqref{eq:radial-constraint-consistency} and~\eqref{eq:radial-constraint-transitivity} ensure that the orders defined for all~$u, v \in V_i \setminus \{\alpha_i^-, \alpha_i^+\}$ by the variables~$\alpha_i^-uv$ and~$\alpha_i^+uv$ are linear.
    Define~$\mathbb A = \{ v \in V_i \setminus \{\alpha_i^-, \alpha_i^+\} \mid \alpha_i^- v \alpha_i^+\}$ and~$\mathbb B = \{ v \in V_i \setminus \{\alpha_i^-, \alpha_i^+\} \mid \alpha_i^- \alpha_i^+ v\}$.
    Let~$\sigma_i^-, \sigma_i^+$ denote the induced cyclic orders.
    We have to show~$\sigma_i^- = \sigma_i^+$.
    To do so, consider pairwise distinct~$u, v, w \in V_i$ and show~$(u, v, w) \in \sigma_i^- \implies (u, v, w) \in \sigma_i^+$.
    Distinguish three cases.
    \begin{enumerate}
        \item $u, v, w \not\in \{\alpha_i^-, \alpha_i^+\}$.
              We may assume $\alpha_i^- u v$, $\alpha_i^- v w$, $\alpha_i^- u w$.
              Distinguish four cases based on how the vertices $u$, $v$ and $w$ are distributed over the sets $\mathbb A$ and $\mathbb B$.
              \begin{enumerate}
                  \item $u, v, w \in \mathbb A$.
                        It is $\alpha_i^- u \alpha_i^+$, $\alpha_i^- v \alpha_i^+$, $\alpha_i^- w \alpha_i^+$.
                        From $\alpha_i^- u v$, $\alpha_i^- u \alpha_i^+$,
$\alpha_i^- v \alpha_i^+$ and constraint~\eqref{eq:radial-constraint-merge-one} we conclude $\alpha_i^+ u v$.
                        Similarly, we conclude $\alpha_i^+ v w$ and $\alpha_i^+ u w$, which then gives us $(u, v, w) \in \sigma_i^+$.
                        \label{itm:8875950323}
                  \item $u, v, \in \mathbb A, w \in \mathbb B$.
                        It is $\alpha_i^- u \alpha_i^+$, $\alpha_i^- v \alpha_i^+$, $\neg \alpha_i^- w \alpha_i^+$.
                        Using constraint~\eqref{eq:radial-constraint-merge-one} we conclude $\alpha_i^+ u v$, $\neg \alpha_i^+ v w$, $\neg \alpha_i^+ u w$, which gives $(u, v, w) \in \sigma_i^+$.
                        \label{itm:9641594092}
                  \item $u, \in \mathbb A, v, w \in \mathbb B$; similar to case~\ref{itm:8875950323}.
                  \item $u, v, w \in \mathbb B$; similar to case~\ref{itm:9641594092}.
              \end{enumerate}
        \item $u, v \not\in \{\alpha_i^-, \alpha_i^+\}, w = \alpha_i^-$
              Distinguish three cases based on how the vertices $u$ and $v$ are distributed over the sets $\mathbb A$ and $\mathbb B$.
              \begin{enumerate}
                  \item $u, v \in \mathbb A$.
                        Then it is $\alpha_i^- u \alpha_i^+$ and $\alpha_i^- v \alpha_i^+$.
                        Constraint~\eqref{eq:radial-constraint-merge-one} then gives $\alpha_i^+ u v$.
                        Constraint~\eqref{eq:radial-constraint-merge-two} gives $\alpha_i^+ v \alpha_i^-$ and $\alpha_i^+ u \alpha_i^-$.
                        This means $(u, v, \alpha_i^- = w) \in \sigma_i^+$.
                  \item $u \in \mathbb A, v \in \mathbb B$.
                        Then it is $\alpha_i^- u \alpha_i^+$, $\alpha_i^- \alpha_i^+ v$ and $\alpha_i^- u v$.
                        Using
constraints~\eqref{eq:radial-constraint-consistency} and~\eqref{eq:radial-constraint-merge-two} we obtain $\alpha_i^+ v \alpha_i^-$ from $\alpha_i^- \alpha_i^+ v$.
                        Using constraint~\eqref{eq:radial-constraint-merge-two} we get $\alpha_i^+ \alpha_i^- u$ from $\alpha_i^- u \alpha_i^+$.
                        Finally, we use constraint~\eqref{eq:radial-constraint-merge-one} to obtain $\alpha_i^+ u v$.
                        This means $(u, v, \alpha_i^- = w) \in \sigma_i^+$.
                  \item $u, v \in \mathbb B$
                        Use constraint~\eqref{eq:radial-constraint-merge-two} to obtain $\alpha_i^+ u \alpha_i^-$ and $\alpha_i^+ v \alpha_i^-$ from $\alpha_i^- \alpha_i^+ u$ and $\alpha_i^- \alpha_i^+ v$, respectively.
                        Then use constraint~\eqref{eq:radial-constraint-merge-one} to obtain $\alpha_i^+ u v$, which means $(u, v, \alpha_i^- = w) \in \sigma_i^+$.
              \end{enumerate}
        \item $u, v \not\in \{\alpha_i^-, \alpha_i^+\}, w = \alpha_i^+$
              Distinguish three cases based on how the vertices $u$ and $v$ are distributed over the sets $\mathbb A$ and $\mathbb B$.
              \begin{enumerate}
                  \item $u, v \in \mathbb A$.
                        Use constraint~\eqref{eq:radial-constraint-merge-two} to obtain $\alpha_i^+ \alpha_i^- v$ and $\alpha_i^+ \alpha_i^- u$ from $\alpha_i^- \alpha_i^+ v$ and $\alpha_i^- \alpha_i^+ u$, respectively.
                        Then use constraint~\eqref{eq:radial-constraint-merge-one} to obtain $\alpha_i^+ u v$, which means $(u, v, \alpha_i^+ = w) \in \sigma_i^+$.
                  \item $u \in \mathbb B, v \in \mathbb A$.
                        Use constraint~\eqref{eq:radial-constraint-merge-two} to obtain $\alpha_i^+ \alpha_i^- v$ and $\alpha_i^+ u \alpha_i^-$ from $\alpha_i^- v \alpha_i^+$ and $\alpha_i^- \alpha_i^+ u$, respectively.
                        Then use constraint~\eqref{eq:radial-constraint-merge-one} to obtain $\alpha_i^+ u v$, which means $(u, v, \alpha_i^+ = w) \in \sigma_i^+$.
                  \item $u, v \in \mathbb B$.
                        Use constraints~\eqref{eq:radial-constraint-consistency}
and~\eqref{eq:radial-constraint-merge-two} to obtain $\alpha_i^+ u \alpha_i^-$ and $\alpha_i^+ v \alpha_i^-$ from $\alpha_i^- \alpha_i^+ u$ and $\alpha_i^- \alpha_i^+ v$, respectively.
                        Then use constraint~\eqref{eq:radial-constraint-merge-one} to obtain $\alpha_i^+ u v$, which means $(u, v, \alpha_i^+ = w) \in \sigma_i^+$.
              \end{enumerate}
        \item $v \not\in \{\alpha_i^-, \alpha_i^+\}, u, w \in \{\alpha_i^-, \alpha_i^+\}$.
              We may assume $u = \alpha_i^-$ and $w = \alpha_i^+$.
              Then $(u, v, w) \in \sigma_i^-$ gives $\alpha_i^- v \alpha_i^+$
and constraint~\eqref{eq:radial-constraint-merge-two} yields $\alpha_i^+ \alpha_i^- v$, which means $(\alpha_i^- = u, v, \alpha_i^+ = w) \in \sigma_i^+$.
    \end{enumerate}

    \noindent Therefore,~$\varphi$ induces well-defined cyclic orders of the vertices on all levels.
    It remains to be shown that no two edges cross.
    Recalling that~$G$ is proper, it is sufficient to show that no two edges~$e, f \in (E_i \cup E_i^- \cup E_i^+)$ cross for~$i = 1, 2, \ldots, k$.

    \begin{enumerate}
        \item $e, f \in E_i$.
              Then constraint~\eqref{eq:radial-constraint-planarity-one} together with the fact that no edge may cross $\varepsilon_i$ implies that $e$ and $f$ do not cross.
        \item $e \in E_i, f \in E_i^-$.
              Let $e = (u, v)$ and $f = (u'', \alpha_{i + 1}^-)$ and suppose $l(f)$ (the case $\neg l(f)$ works symmetrically).
              Then constraint~\eqref{eq:radial-constraint-planarity-four} ensures that $e$ and $f$ do not cross.
              \label{itm:6268188379}
        \item $e \in E_i, f \in E_i^+$; works symmetrically to
case~\ref{itm:6268188379} with constraint~\eqref{eq:radial-constraint-planarity-three}.
        \item $e \in E_i^-, f \in E_i^+$.
              Then constraint~\eqref{eq:radial-constraint-planarity-two} ensures that $e$ and $f$ are embedded locally to the left and right of $\varepsilon_i$, respectively, or vice versa.
              Together with the fact that no edge may cross $\varepsilon_i$, this means that $e$ and $f$ do not cross.
        \item $e, f \in E_i^-$ or $e, f \in E_i^+$.
              Because $e$ and $f$ share an endpoint they do not cross.
    \end{enumerate}

    \noindent Therefore, no two edges cross, which means that~$\varphi$ induces a radial level planar drawing.
\end{proof}

\section{Construction of $\boldsymbol{\varphi^+}$}
\label{sec:phiplus}

\existencephiplus*

For the proof of this result, we assume that~$\varphi$ is a satisfying
assignment for~$\mathcal S(G,A^+,A^-)$.  We now derive a truth
assignment $\varphi^+$ for $\mathcal S(G^+,B^+,B^-)$ from $\varphi$.
Afterwards, we show that $\varphi^+$ satisfies~$\mathcal S(G^+,B^+,B^-)$.

Let $e=uu',f=vv'$ be two edges between level $i$ and level $i+1$ of
$G$ and let~$\eps_i = (\alpha_i^+,\alpha_{i+1}^-)$ denote the
reference edge between these levels.  We introduce a function
$\psi(\eps_i,e,f)$ to deduce the order of the edges~$\eps_i,e,f$ in a
drawing that corresponds to~$\varphi$.
\[
\psi(\eps_i,e,f)=\begin{cases}
		\phi(\alpha_i^+uv)&, \text{ if $\alpha_i^+,u,v$ are pairwise distinct}\\ 
		\phi(\alpha_{i+1}^-u'v')&, \text{ if $\alpha_{i+1}^-,u',v'$ are pairwise distinct}\\ 
		\phi(\l(f))&, \text{ if $\alpha_i^+=v\ne u$ or $\alpha_{i+1}^-=v'\ne u'$}\\
                \neg\phi(\l(e))&, \text{ if $\alpha_i^+=u\ne v$ or $\alpha_{i+1}^-=u'\ne v'$}
	\end{cases}
\]

Note that if~$\alpha_i^+,u,v$ and~$\alpha_{i+1}^-,u',v'$ are pairwise
distinct, then by Eq.~(\ref{eq:radial-constraint-planarity-one}) it
is~$\phi(\alpha_i^+uv) = \phi(\alpha_{i+1}^-uv)$.  Similarly,
if~$\alpha_i^+ = v$ and~$\alpha_{i+1}^- = u'$ (or~$\alpha_i^+=u$
and~$\alpha_{i+1}^-=v'$), then by
Eq.~(\ref{eq:radial-constraint-planarity-two}) it
is~$\phi(\l(f)) = \neg(\phi(\l(e))$.  Therefore~$\psi$ is well-defined.

Based on this, we can now define the orderings of triples of
subdivision vertices of~$G^+$, which leads to an
assignment~$\varphi^+$ for~$\mathcal S(G^+,B^+,B^-)$.

To this end, we define a mapping~$O \colon V(G^+) \to V$ that maps each vertex of~$V(G^+)$ to a vertex in~$G$.
For each vertex~$v$ of~$G^+$ that is part of a stretch edge~$e(w)$ of~$G^\star$ for some vertex~$w$ of~$G$, we set~$O(v) = w$.
For an example, see the orange vertex~$w$ of~$G$ in Fig.~\ref{fig:o}.
The encircled orange vertices in~$G^+$ are mapped to~$w$.
It remains to define~$O(v)$ for vertices~$v \in V(G^+)$ that are subdivision vertices of an edge~$xx'$ of~$G^\star$ that its not a stretch edge.
Then it is~$L(x') = L(x)+1$, and we map~$v$ to the vertex~$w$ with~$x = \tops(w)$ if~$L(v) = L(x)$ and to the vertex~$w'$ with~$x' = \bots(w')$, otherwise.
For an example, see the edge~$(x, x')$ of~$G$ in Fig.~\ref{fig:o}.
The encircled purple (green) subdivision vertices in~$G^+$ are mapped to~$x$~($x'$).

\begin{figure}[t]
    \centering
    \includegraphics{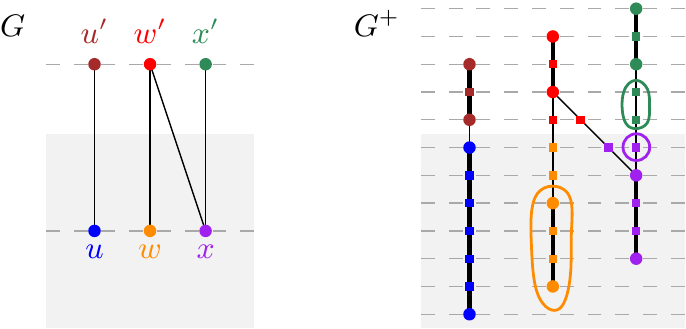}
    \caption{
        Definition of the mapping $O \colon V(G^+) \to V$.
        Vertices of $G^+$ are mapped to the vertex of $G$ with the matching color.
    }
    \label{fig:o}
\end{figure}

We are now ready to define the assignment~$\varphi^+$
for~$\mathcal S(G^+,B^+,B^-)$ for a satisfying assignment~$\varphi$
of~$\mathcal S(G,A^+,A^-)$.  Let $x,y,z\in V_j$ be three distinct
vertices on level $j$ with $L(j) = i$ and such that
$x\in\{\beta_j^-,\beta_j^+\}$.  Observe that
$O(x) \in \{\alpha_i^+,\alpha_{i+1}^-\}$.
If the vertices~$O(x)$, $O(y)$, $O(z)$ are pairwise distinct, we
set~\[\phi^+(xyz) = \phi\left(O(x)O(y)O(z)\right)\,.\]

Otherwise, two of the corresponding vertices are the same. Hence, two
vertices of $x,y,z$ are subdivision vertices of original edges in $G$.
Note that $x$ can by definition of $G^\star$ not be such a subdivision
vertex. Hence, vertex $y$ is a subdivision vertex of an edge $e\in G$
and vertex $z$ is a subdivision vertex of an edge $f\in G$, both of
which connect a vertex of level $i$ to a vertex on level $i+1$. Then
we set \[\phi^+(xyz)=\psi(\eps_i,e,f)\,,\]
where~$\eps_i = (\alpha_i^+,\alpha_{i+1}^-)$ is the reference edge
between these levels. 
Moreover, we set~$\phi^+(\l(e)) = \phi(\l(O(u),O(v)))$ for each edge $e=(u,v)$
of~$G^+$ with $u=\beta_j^+$ or~$v= \beta_{j+1}^-$ but not both for
some level $j$ of $G^+$.

\begin{lemma}
  The assignment $\varphi^+$ satisfies~$\mathcal S(G^+,B^+,B^-)$.
\end{lemma}

\begin{proof}We check that~$\varphi^+$ satisfies each part of the
  constraint system.  We denote by~$V_j$ the vertices of~$G^+$ on
  level $j$.
  \begin{description}
  \item[$\boldsymbol{\phi^+}$ satisfies $\boldsymbol{\mathcal L'_{G^+}(\beta_j^+)}$.]
  Let $\beta_j^+\in B^+$. We aim to show that $\phi^+$ satisfies
  $\mathcal L'_{G^+}(\beta_j^+)$.  Let
  $y,z \in V_j\setminus\{\beta_j^+\}$ be distinct.  If
  $O(\beta_j^+),O(y),O(z)$ are distinct, then we have
  $\phi^+(\beta_j^+yz)=\phi(O(\beta_j^+)O(y)O(z))$ as well as
  $\phi^+(\beta_j^+zy)=\phi(O(\beta_j^+)O(z)O(y))$.  Note that
  $O(\beta_j^+)\in \{\alpha_{L(j)}^+,\alpha_{L(j)}^-\}$.  Since $\phi$
  satisfies $\mathcal L'_G(\alpha_{L(j)}^+)$ and
  $\mathcal L'_G(\alpha_{L(j)}^-)$, we obtain
  $\phi^+(\beta_j^+yz)=\neg \phi^+(\beta_j^+zy)$.

  If $O(\beta_j^+),O(y),O(z)$ are not distinct, then we must have
  $O(y)=O(z)$, since $\beta_j^-$ cannot be a subdivision vertex of an
  edge that is not a stretch edge.  Hence, the vertices $y$, $z$ are
  subdivision vertices of edges $e,f$ of $G^\star$ connecting vertices
  on levels $L(j)$ and $L(j)+1$.  Note that this also implies
  $O(\beta_j^+) \ne O(y)$. We obtain
  $\phi^+(\beta_j^-,y,z)=\psi(\eps_i,e,f)$ and
  $\phi^+(\beta_j^-,z,y)=\psi(\eps_i,f,e)$. Since
  $O(y) = O(z)\ne O(\beta_j^-)$, we have with the definition of
  $\psi$ that $\psi(\eps_i,e,f)= \neg \psi(\eps_i,f,e)$. This yields
  $\phi^+(\beta_j^-,y,z)=\neg \phi^+(\beta_j^-,z,y)$.

\item[$\boldsymbol{\phi^+}$ satisfies
  $\boldsymbol{\mathcal L'_{G^+}(\beta_j^-)}$.] This case is analogous
  to the previous case.

\item[$\boldsymbol{\phi^+}$ satisfies
  $\boldsymbol{\mathcal C_{G^+}(\beta_j^+,\beta_j^-)}$.]  We next show
  that $\phi^+$ satisfies the constraints
  $\mathcal C_{G^+}(\beta_j^+,\beta_j^-)$ for all $j$ with
  $\beta_j^-\ne \beta_j^+$.  By definition of~$G^\star$ and~$B^+,B^-$,
  these constraints are non-trivial only for levels $j=m_i$ for a
  level $i$ of $G$ with $\alpha_i^+\ne\alpha_i^-$.
 
  Note that due to the construction of~$G^\star$, level $m_i$ is only
  crossed by stretch edges, which implies that the restriction of~$O$
  to the vertices of~$G^+$ on level $m_i$ is
  injective. Moreover,~$O(\beta_j^+) = \alpha_i^+$
  and~$O(\beta_j^-) = \alpha_i^-$.  Hence, the triple~$x,y,z \in V_j$
  of distinct vertices map injectively to triples of distinct vertices
  of~$G$ on level $i$.  Since the value of~$\varphi^+(xyz)$ is defined
  in terms of~$\varphi(O(x)O(y)O(z))$, it follows that~$\varphi^+$
  satisfies Eq.~(\ref{eq:radial-constraint-merge-one})
  and~(\ref{eq:radial-constraint-merge-two}) if~$\varphi$ does.

\item[$\boldsymbol{\phi^+}$ satisfies
  $\boldsymbol{\mathcal P_{G^+}(\delta_j)}$.]  We
  finally show that $\phi^+$ satisfies the constraints
  $\mathcal P_{G^+}(\delta_j)$ for any two consecutive
  levels $j$ and $j+1$ of $G^+$.  We distinguish two cases, based on
  whether $L(j)=L(j+1)$.

 \noindent \textbf{Case 1.} $L(j)=L(j+1)=i$.
 First observe that, except for $(\beta_j^+,\beta_{j+1}^-)$, the
 reference vertex $\beta_j^+$ has no outgoing edges and
 $\beta_{j+1}^-$ has no incoming edges.  Namely, vertex $\beta_j^+$
 can only have more outgoing edges, if
 $\beta_j^+=\tops(\alpha_{i}^+)$, and vertex $\beta_{j+1}^-$ can only
 have more incoming edges, if $\beta_{j+1}^-=\bots(\alpha_i^-)$.  But
 $\bots(\alpha_i^-)$ and $\tops(\alpha_i^+)$ occupy the first and the
 last level of $G^+$ corresponding to level $i$ of $G$.  Hence this is
 not possible if $L(j)=L(j+1)$. Therefore the
 equations~(\ref{eq:radial-constraint-planarity-two})--(\ref{eq:radial-constraint-planarity-four})
 are trivially satisfied.

 Now consider two independent edges $(y,y'),(z,z')$ between levels $j$
 and $j+1$ that are different from $(\beta_j^+,\beta_{j+1}^-)$.
 
 Since the vertices $y$, $y'$ are adjacent in $G^+$, they either both
 subdivide the same edge of $G^\star$, or one of them is a vertex of
 $G^\star$ and the other one subdivides an incident edge.  As
 $L(j)=L(j+1)$, we obtain $O(y)=O(y')$ in both cases.  Similarly, we
 obtain $O(z)=O(z')$.  If $j\ge m_i$, then we have
 $O(\beta_j^+)=O(\beta_{j+1}^-)=\alpha_i^+$.  If $j < m_i$, then we
 have $O(\beta_j^+)=O(\beta_{j+1}^-)=\alpha_i^-$.  In both cases we
 have $O(\beta_j^+)=O(\beta_{j+1}^-)$.

 Observe that~$O(y) = O(\beta_j^+) = \alpha_i^-$ implies that
 $j <m_i$, and hence~$y$ is a subdivision vertex of an outgoing edge
 of~$\bots(\alpha_i^-)$ in~$G^\star$.  This is, however, impossible
 since the only outgoing edge of~$\bots(\alpha_i^-)$ is the stretch
 edge of~$\alpha_i^-$.  Symmetrically,
 $O(y) = O(\beta_j^+) = \alpha_i^+$ implies $j\ge m_i$, and hence that
 $y$ is a subdivision vertex of an incoming edge
 of~$\tops(\alpha_i^+)$, which is again impossible.  Thus we find that
 $O(\beta_j^+) \ne O(y)$.  Likewise, it is~$O(\beta_j^+) \ne O(z)$.

 If~$O(y) \ne O(z)$, then the vertices~$O(\beta_j^+) = O(\beta_j^-)$,
 $O(y) = O(y')$ and~$O(z) = O(z')$ are pairwise distinct.  Thus~$O$
 maps the triples $t_1 = \beta_j^+yz$ and~$t_2 = \beta_{j+1}^-y'z'$ to
 the same triple $t$.  Since~$\phi^+(t_i) = \phi(t)$ for~$i=1,2$, it
 follows that~$\phi^+(t_1) = \phi^+(t_2)$, i.e.,~$\varphi^+$ satisfies
 Eq.~(\ref{eq:radial-constraint-planarity-one}).

 If $O(y)=O(z)$, then $O(y')=O(z')$, and $(y,y'),(z,z')$ originate
 from edges $e,f \in G$.  Now, if~$j<m_i$, then $e$ and~$f$ connect
 vertices on level~$i-1$ to vertices on level $i$.  In this
 case,~$\phi^+(\beta_j^+ef)$ and~$\phi^+(\beta_j^-ef)$ are both
 defined in terms of~$\psi(\eps,e,f)$,
 where~$\eps = (\alpha_{i-1}^+,\alpha_i^-)$.  If~$j \ge m_i$, then $e$
 and $f$ connect vertices on levels $i$ and~$i+1$.  Then~$\phi^+$ of
 both triples is defined as~$\psi(\eps,e,f)$
 for~$\eps = (\alpha_i^+,\alpha_{i+1}^-)$.  In both case
 Eq.~(\ref{eq:radial-constraint-planarity-one}) is satisfied.

 \noindent \textbf{Case 2.} $i=L(j)<L(j+1)=i+1$.  In this case we have
 $\beta_j^+=\tops(\alpha_i^+)$ and
 $\beta_{j+1}^-=\bots(\alpha_{i+1}^-)$. Let $xx'\in G^+$ be any edge
 between level $j$ and level $j+1$.  Since~$L(j) \ne L(j+1)$,
 $e = (O(x),O(x'))$ is an edge of~$G$.  Further, we obtain
\begin{alignat*}{7}
 &(x,x')&&\in E_i(G^+)   &&{} \Leftrightarrow e\in E_i&&\\
 &(x,x')&&\in E_i^+(G^+) &&{} \Leftrightarrow x=\beta_j^+
&&{}\Leftrightarrow v(x)&&{}=\alpha_i^+&&{}\Leftrightarrow e\in E_i^+\\
 &(x,x')&&\in E_i^-(G^+) &&{} \Leftrightarrow x'=\beta_{j+1}^-
&&{}\Leftrightarrow v(x')&&{}=\alpha_{i+1}^-&&{}\Leftrightarrow e\in E_i^-\,.
\end{alignat*}

Let $yy',zz'$ be distinct edges between levels $j$ and $j+1$ in $G^+$
that are different from $(\beta_j^+,\beta_{j+1}^-)$.  We have that
$e = (O(y),O(y'))$ and $f = (O(z),O(z'))$ are edges of~$G$.
We further distinguish case based on whether $e$ and $f$ are
independent.  Assume $e,f$ are independent.

\textbf{Case2a.} $e$ and $f$ are independent.  
If $e,f\in E_i$, then we have $(y,y'),(z,z') \in E_i(G^+)$ and by
definition of~$\phi^+$ it is
$\phi^+(\beta_j^+yz)=\phi(\alpha_i^+O(y)O(z))$ and
$\phi^+(\beta_{j+1}^-y'z')=\phi(\alpha_{i+1}^-O(y')O(z'))$.  Since
$\phi$ satisfies $\mathcal P_G(\alpha_i^+, \alpha_{i+1}^-)$, it follows
that $\phi^+$ satisfies Eq.~(\ref{eq:radial-constraint-planarity-one})
for $G^+$ and edges $(y,y'), (z,z')$.

If we have $e \in E_i^+$ and $f\in E_i^-$, then we have
$(y,y') \in E_j^+(G^+)$ and $(z,z') \in E_i^-(G^+)$.  The, by
definition of~$\phi^+$, it is
$\phi^+(\l(y,y'))=\phi(\l(\alpha_i^+O(y'))$ and
$\phi^+(\l(z,z'))=\phi(\l(O(z)\alpha_{i+1}^-))$.  Since $\phi$
satisfies $\mathcal P_G(\eps_i)$, it follows that
$\phi^+$ satisfies Eq.~(\ref{eq:radial-constraint-planarity-two}) for
$G^+$ and edges $(y,y'),(z,z')$.

If we have $e\in E_i^+$ and $f\in E_i$, then we have
$(y,y') \in E_i^+(G^+)$ and $(z,z') \in E_i(G^+)$.  By the definition
of~$\phi^+$ it is
$\phi^+(\l(\beta_j^+,y')) = \phi(\l(\alpha_i^+,O(y')))$ and
$\phi^+(\beta_{j+1}^-z'y') = \phi(\alpha_{i+1}^-O(z')O(y'))$.  Since
$\phi$ satisfies $\mathcal P_G(\eps_i)$, it follows
that $\phi^+$ satisfies
Eq.~(\ref{eq:radial-constraint-planarity-three}) for $G^+$ and edges
$(y,y'),(z,z')$.  The case $e\in E_i^-$ and $y\in E_i$ can be argued
analogously.

\noindent\textbf{Case 2b.} $e,f$ are dependent.
If $O(y)=O(z)$, then $y$ and~$z$ are subdivision vertices of two edges
$e'$, $f'$ in~$G^\star$ that share their source, which hence lies on a
level strictly below~$j$.  In particular, it
is~$y,z \ne \tops(\alpha_i^+) = \beta_j^+$.  Moreover, it
is~$O(y') \ne O(z')$.  Thus, by definition, we have
$\phi^+(\beta_j^+yz)=\phi^+(\beta_{j+1}^-y'z') = \psi(\eps_i,e,f)$,
where~$\eps_i = (\alpha_i^+,\alpha_{i+1}^-)$.

If the vertices $O(\beta_j^+) = \alpha_{i+1}^-$, $O(y')$, $O(z')$ are
pairwise distinct, then we have $(y,y'),(z,z') \in E_i(G^+)$.  Then
$\phi^+(\beta_j^+yz)=\phi(\alpha_{i+1}^-O(y')O(z'))$ and
$\phi^+(\beta_{j+1}^-y'z')=\phi(\alpha_{i+1}^-O(x')O(y'))$.  Therefore
$\phi^+$ satisfies Eq.~(\ref{eq:radial-constraint-planarity-one}) for
$G^+$ and edges $(y,y')$ and~$(z,z')$.

If $\alpha_{i+1}^-=O(y')\ne O(z')$, then we have
$(y,y') \in E_i^+(G^+)$ and $(z,z') \in E_i(G^+)$. Then
$\phi^+(\beta_j^+yz)=\phi(\l(e))$ and
$\phi^+(\beta_{j+1}^-y'z')=\phi(\l(e))$.  Therefore $\phi^+$ satisfies
Eq.~(\ref{eq:radial-constraint-planarity-three}) for $G^+$ and edges 
$(y,y)',(z,z')$.

Analogously, we obtain for $\alpha_{i+1}^-=O(z')\ne O(y')$ that
$\phi^+$ satisfies Eq.~(\ref{eq:radial-constraint-planarity-three})
for $G^+$ and edges $(y,y')$ and $(z,z')$.

Finally, the case~$O(y') = O(z')$ can be handled analogously to the
case~$O(y) = O(z)$.
 \end{description}
 
 This concludes the proof that~$\varphi^+$ satisfies
 $\mathcal S(G^+,B^+,B^-)$.
\end{proof}

\section{Even Crossings Criterion}
\label{sec:radial-drawing}

\threecurves*

\begin{proof}
  If there are no crossings, then both sides hold.  Hence, assume there is
  at least one crossing.  By perturbations we achieve that every
  crossing has another distance to the center $m$. We order the
  crossings by distance to $m$.  We add a concentric circle $C_X^Y$
  between any two consecutive crossings $X,Y$ such that $C_X^Y$
  intersects $a,b,c$ each once.
  
  Then, the order of the intersection points of $a,b,c$ must change
  between every two consecutive circles. Thus, that order is the same
  in $C_1$ and $C_2$, if and only if we added an odd number of
  circles. This in turn holds if and only if we have an even number of
  crossings.
\end{proof}

\section{$\phi$ from Hanani-Tutte Drawing}
\label{sec:phi-from-HTdrawing}

\phiFromDrawing*

For the proof of the lemma we first construct the assignment~$\varphi$
and then show that it satisfies~$\mathcal S(G^+,A^+,A^-)$. 
For easier readable formulations we set $\text{true}=0$ and $\text{false}=1$. %

Let $G^\star=(V,E)$ have $k$ levels and with a radial Hanani-Tutte drawing $\Gamma$.
Let $G^+$ be the level graph obtained by subdividing all edges such that $G^+$ is proper.
Let $G^+$ have $\kappa$ levels.

For three distinct vertices $x$, $y$, $z$ on level $j$ with $x=\alpha_j^-$ or
$x=\alpha_j^+$, we set $\psi(xyz)=0$ if and only if $x$, $y$, $z$ appear clockwise on the circle representing level~$j$.

If two edges $e$, $f$ are adjacent in a vertex $v$ with $\cross(e,f)\equiv 1\mod
2$, then we say they have a \emph{phantom crossing} at $v$. We denote by $\pcr$
the function counting crossings and additionally adding $1$ if there is a phantom crossing. 
With the phantom crossings, any two edges in $G$ cross an even number of times, even if they are not independent.
For any edge $e=uv$ in $G^+$ between level $j$ and $j+1$ with $u=\alpha_j^+$ or
$v=\alpha_{j+1}^-$ we set $\psi(\l(e))=0$ if and only if $e$ is locally left of $\eps_j$.
We further set $\phi(\l(e))\equiv \psi(\l(e)) + \cross(e,\eps_j) \mod 2$ to
switch that value in case of a phantom crossing.

Let $v \in V(G^+)$ be a vertex. If $v$ is a subdivision vertex of an edge $e$, then we set $e(v)=e$. Otherwise we set $e(v)=\emptyset$. We say $\emptyset$ has no crossings with anything but stretches over all levels.
This helps to avoid case distinctions.
For an edge $e=(u,v)$ of $G$ we write $e^j$ for the subdivision path of $e$ that starts at $u$ and ends in level $j$. We set $\emptyset^j=\emptyset$.

Let $x,y,z\in V_j(G^+)$ be disjoint with $x=\alpha_j^-$ or $x=\alpha_j^+$.
We set 
$$\phi(xyz)\equiv\psi(xyz)+\pcr(e(x)^j,e(y)^j,e(z)^j)\mod 2\,.$$
We thereby switch the order of $x$, $y$, $z$ if and only if at least two of them
are subdivision vertices and the corresponding edges cross an odd number of
times up to level $j$.
This finishes the construction of~$\varphi$.

\begin{lemma}
  The assignment~$\varphi$ satisfies~$\mathcal S(G^+,A^+,A^-)$.
\end{lemma}

\begin{proof}
 First note that $\psi$ satisfies $\mathcal L(\alpha_j^+), \mathcal
L(\alpha_j^-)$ and $\mathcal C(\alpha_j^+,\alpha_j^-)$ for $1\le j\le\kappa$.
\begin{description}
\item[$\boldsymbol\varphi$ satisfies
  $\boldsymbol{\mathcal L(\alpha_j^+), \mathcal L(\alpha_j^-)}$.]  For
  three distinct vertices $x,y,z \in V_j(G^+)$ we have by definition of $\phi$
that $\phi(xyz)+\phi(xzy) \equiv \psi(xyz)+\psi(xzy) \mod 2$.
  Since $\psi$ satisfies Eq.~\eqref{eq:radial-constraint-consistency}, so does
$\varphi$.%

  \item[$\boldsymbol\varphi$ satisfies $\boldsymbol{\mathcal C(\alpha_j^+,\alpha_j^-)}$.]
   For Eq.~\eqref{eq:radial-constraint-merge-one} let $1\le j\le \kappa$ such that
    $\alpha_j^-\ne \alpha_j^+$. Let $u,v\in V_j(G^+)\setminus\{\alpha_j^-,\alpha_j^+\}$ be distinct.  By
    definition of $\phi$ we have that
  \begin{equation*}
\begin{split}
    & (\phi(\alpha_j^-uv) + \phi(\alpha_j^+uv)) + (\phi(\alpha_j^-u\alpha_j^+) +
\phi(\alpha_j^-v\alpha_j^+))\\
     \equiv~& (\psi(\alpha_j^-uv) + \psi(\alpha_j^+uv)) +
(\psi(\alpha_j^-u\alpha_j^+) +
\psi(\alpha_j^-v\alpha_j^+))\\
  & + 2\big(\pcr(e(\alpha_j^-),u) + \pcr(e(\alpha_j^-),v)
+\pcr(e(\alpha_j^+),u) + \pcr(e(\alpha_j^+),v)\\
 & +\pcr(u ,v)+ \pcr(e(\alpha_j^-),e(\alpha_j^+))\big)\\
 \equiv{}& (\psi(\alpha_j^-uv) + \psi(\alpha_j^+uv)) +
(\psi(\alpha_j^−u\alpha_j^+) +
\psi(\alpha_j^-v\alpha_j^+)) \mod 2
\end{split}
\end{equation*}

	Thereby Eq.~\eqref{eq:radial-constraint-merge-one} holds.

       For Eq.~\eqref{eq:radial-constraint-merge-two} let
    $u\in V_j(G^+)\setminus\{\alpha_j^+,\alpha_j^-\}$.  
    With the definition of $\phi$ we obtain
    \[\phi(\alpha_j^-u\alpha_j^+)+\phi(\alpha_j^+\alpha_j^-u)\equiv
    \psi(\alpha_j^-u\alpha_j^+)+\psi(\alpha_j^+\alpha_j^-u)\mod 2\,.\]
    Thereby Eq.~\eqref{eq:radial-constraint-merge-two} holds.

       \end{description}
      
      \noindent It remains to argue that~$\varphi$
      satisfies~$\mathcal P(\eps_j)$ for all levels.  We first make
      the following observation.  Let $(u,u')$ be an edge between
      level $j$ and $j+1$ in $G^+$ that is an original edge in
      $G$. Let $f$ be another original edge in $G$. Since $(u,u'),f$
      are monotone and original edges cross an even number of times,
      we then obtain
	\begin{equation}
		\pcr(f_j^{j+1},uu')=\pcr(f,uu')\equiv 0 \mod 2	\label{eq:sectionCr}
	\end{equation}
        We now show that~$\varphi$ satisfies the individual equations
        of~$\mathcal P(\eps_j)$.

\begin{description}
\item[$\boldsymbol{\varphi}$ satisfies
  Eq.~\eqref{eq:radial-constraint-planarity-one}
  of~$\boldsymbol{\mathcal P(\eps_j)}$.]  Let $1\le j \le \kappa - 1$
  and let $(u,u'),(v,v')\in E_j(G^+)$ be independent.  Then we argue
  that
 \begin{equation}
  \begin{split}
     &{}\phi(\alpha_j^+uv) + \phi(\alpha_{j+1}^-u'v')\\
     \equiv{}&{}\psi(\alpha_j^+uv) + \psi(\alpha_{j+1}^-u'v')
       + \pcr((u,u'),(v,v'),\eps_j)\label{eq:P1switch}
 \end{split}
  \end{equation}
This means we change the order of the ends of $(u,u'),(v,v'),\eps_j$ on exactly one of the levels $j$ and $j+1$ if and only if they cross an odd number of times.
	With Lemma~\ref{lem:equalOrderImpliesEvenCrossNumberOf3Edges} we then obtain
that Eq.~\eqref{eq:radial-constraint-planarity-one} holds for  $(u,u'),(v,v')$.

	With the definition of $\phi$ it suffices to show the following three statements:
\begin{alignat}{3}
		  & \pcr(e(\alpha_{j+1}^-)^{j+1},e(u')^{j+1})  &&  +
\pcr(e(\alpha_j^+)^j,e(u)^j)  && \equiv \pcr(\eps_j, (u,u')) \label{eq:cr_u} \\		  
      & \pcr(e(\alpha_{j+1}^-)^{j+1},e(v')^{j+1})  &&
+\pcr(e(\alpha_j^+)^j,e(v)^j)  && \equiv \pcr(\eps_j, (v,v')) \label{eq:cr_v} \\
		  & \pcr(e(u')^{j+1},e(v')^{j+1})              &&  + \pcr(e(u)^j,e(v)^j)
\label{eq:cr_uv}
 && \equiv \pcr((u,u'),(v,v'))
  \end{alignat}

  We show Eq.~\eqref{eq:cr_uv}. Eq.~\eqref{eq:cr_u} and  Eq.~\eqref{eq:cr_v} can be shown analogously
noting that $\eps_j=(\alpha_j^+,\alpha_{j+1}^-)$ is independent from
$(u,u'),(v,v')$.

	\textbf{Case 1. } $u$, $u'$ or $v$, $v'$ are subdivision vertices.
	Then we obtain $e(u)=e(u')$ or $e(v)=e(v')$. In both cases the left side
reduces with \eqref{eq:sectionCr} to $\pcr((u,u'),(v,v'))$.

	\textbf{Case 2. } $u$, $u'$ or $v$, $v'$ are original vertices.
	Then we have $e(u)=e(u')=\emptyset$ or $e(v)=e(v')=\emptyset$.
	Then the left side equals $0$. By Eq.~\eqref{eq:sectionCr} the right side is also even.

	\textbf{Case 3. } $u$, $v$ are original vertices and $u'$, $v'$ are subdivision vertices. Then  we have $e(u)=e(v)=\emptyset$ and $e(u')^{j+1}=uu'$ and $e(v')^{j+1}=vv'$ and the equivalence holds. 

	\textbf{Case 4. } $u'$, $v'$ are original vertices and $u$, $v$ are subdivision
vertices. Noting $\pcr(e(u)^j,e(v)^j)+\pcr(e(u)_j^\kappa,e(v)_j^\kappa)=\pcr(e(u),e(v))\equiv 0\mod 2$ we can argue analogously.

	\textbf{Case 5. } $u$, $v'$ are original vertices and $u'$, $v$ are subdivision vertices. Then we have $e(u)=e(v')=\emptyset$ and the left side equals $0$. Further, we have\\ $\pcr((u,u'),(v,v'))=\pcr(e(u'),e(v))\equiv 0\mod 2$. I.e., both sides are even.

	\textbf{Case 6. } $u',v$ are original vertices and $u,v'$ are subdivision vertices. Then we can argue analogously.
  Hence, we have that $\phi$ satisfies Eq.~\eqref{eq:radial-constraint-planarity-one}.

\item[$\boldsymbol{\varphi}$ satisfies Eq.~\eqref{eq:radial-constraint-planarity-two} of~$\boldsymbol{\mathcal P(\eps_j)}$.] Let $e=(u,u')\in E_j^+(G^+)$ and let $f=(v,v')\in E_j^-(G^+)$.
   Then we adjust the drawing by perturbing possible phantom crossings of $e,f$
with $\eps_j$ in $\alpha_j^+$ and $\alpha_{j+1}^-$ to the space between circle $j$ and circle $j+1$.
  Thereby, the states of $e,f$ being locally left of $\eps_j$ change if and only if they have a phantom crossing. This is the case if and only if $\psi$ and $\phi$ differ for $\l$ of the corresponding edge. I.e., $\phi(\l(e))$ and $\phi(\l(f))$ describe, whether $e,f$ are locally left of $\eps_j$.
  Consider the closed curve $c$ that is the union of $c(e),c(\eps_j)$ and the part $d$ of circle $j+1$ between $u'$ and $\alpha_{j+1}^-$, such that circle $j$ is on the outside of $c$.
 Then the edge $f$ is outside of $c$ at circle $j$. Note that $\eps_j$ has to be an original edge.
 Since $f$ crosses $e,\eps_j$ an even number of times each (as their
corresponding original edges cross only between level $j$ and $j+1$), and it
does not cross $d$, edge $f$ has to be outside of $c$ at $\eps$-close distance
to circle $j+1$. We thereby obtain $\phi(\l(e))=1-\phi(\l(f))$.  
 
\item[$\boldsymbol{\varphi}$ satisfies
Eq.~\eqref{eq:radial-constraint-planarity-three} of~$\boldsymbol{\mathcal
P(\eps_j)}$.] Let $(u,u')\in E_j^+$ and let $(v,v')\in E_j$ be independent. 
Note that $e(\alpha_{j+1}^-)^{j+1})=\eps_j$ and that $e(u')^{j+1}=(u,u')$ and
likewise for $v'$ since $\alpha_{j+1}^-=u'$ must be an original vertex. 
Let circle $a$ be a circle that is $\eps$-close outside of circle $j$.

 Then we have 
  \begin{align*}
    & \phi(\alpha_{j+1}^-v'u') -  \psi(\alpha_{j+1}^-v'u')\\
 & \equiv
\pcr(e(v')^{j+1},e(u')^{j+1})+\pcr(e(\alpha_{j+1}^-)^{j+1},e(v'))+\pcr(e(\alpha_{j+1}^-)^{j+1},e(u')) \\
 & \equiv  \pcr((v,v'),(u,u'))+\pcr(\eps_j,(v,v'))+\pcr(\eps_j,(u,u'))\\
 & \equiv  \pcr(\eps_j,(v,v'),(u,u'))\\
 & \equiv  \cross(\eps_j,(v,v'),(u,u'))+\pcr((u,u')^a,(v,v')^a,\eps_j^a)\\
 & \equiv  (\psi(\l(u,u'))-\psi(\alpha_{j+1}^-v'u'))+\pcr((u,u')^a,(v,v')^a,\eps_j^a)\\
 & \equiv  \phi(\l(u,u')) - \psi(\alpha_{j+1}^-v'u')\mod 2\,.
  \end{align*}

We obtain $\phi(\l(u,u'))= \phi(\alpha_{j+1}^-v'u')$.
\item[$\boldsymbol{\varphi}$ satisfies Eq.~\eqref{eq:radial-constraint-planarity-four} of~$\boldsymbol{\mathcal P(\eps_j)}$.] Can be argued similarly to the previous case.
\end{description}
\end{proof}

\end{document}